\documentclass[11pt]{article}

\usepackage[english,ngerman]{babel}
\usepackage{scrhack}
\usepackage{srcltx}

\usepackage{verbatim}

\usepackage{enumerate}

\usepackage{mathrsfs}
\usepackage{amssymb}
\usepackage{bbm}
\usepackage{color}
\usepackage{graphicx}
\usepackage{float}
\usepackage{mathtools}
\usepackage{amsthm}
\usepackage{amsmath}

\newcommand{\norm}[1]{\left\|#1\right\| }

\newcommand{\epsi}[0]{\varepsilon}
\newcommand{\mrm}[1]{{\mathrm{#1}}}
\newcommand{\field}[1]{\mathbb{#1}}

\newcommand{\R}{\field{R}}
\newcommand{\N}{\field{N}}
\newcommand{\Z}{\field{Z}}
\newcommand{\id}[0]{\mathbf{1}}

\newcommand{\Or}{{\mathcal{O}}}
\newcommand{\E}{{\mathrm{e}}}

\newcommand{\D}{{\mathrm{d}}}

\DeclareMathOperator{\Tr}{Tr}

\newcommand{\im}{\mathrm{i}}

\newcommand{\ph}{\varphi}
\newcommand{\tube}{\mathcal{T}}

\newcommand{\weq}[2]{\stackrel{\mathclap{#1}}{#2}}

\theoremstyle{plain}
\newtheorem{thm}{Theorem}
\theoremstyle{plain}
\newtheorem{lem}{Lemma}[section]
\newtheorem{prop}[lem]{Proposition}
\newtheorem{cor}[lem]{Corollary}
\newtheorem*{thm*}{Theorem}%
\newtheorem*{lem*}{Lemma}
\newtheorem*{prop*}{Proposition}
 
\newtheorem*{cor*}{Corollary}

\newtheorem{defn}{Definition}[section]
\newtheorem*{defn*}{Definition}
\theoremstyle{remark}
\newtheorem{rem}{Remark}

\makeatletter
\DeclareFontFamily{OMX}{MnSymbolE}{}
\DeclareSymbolFont{MnLargeSymbols}{OMX}{MnSymbolE}{m}{n}
\SetSymbolFont{MnLargeSymbols}{bold}{OMX}{MnSymbolE}{b}{n}
\DeclareFontShape{OMX}{MnSymbolE}{m}{n}{
    <-6>  MnSymbolE5
   <6-7>  MnSymbolE6
   <7-8>  MnSymbolE7
   <8-9>  MnSymbolE8
   <9-10> MnSymbolE9
  <10-12> MnSymbolE10
  <12->   MnSymbolE12
}{}
\DeclareFontShape{OMX}{MnSymbolE}{b}{n}{
    <-6>  MnSymbolE-Bold5
   <6-7>  MnSymbolE-Bold6
   <7-8>  MnSymbolE-Bold7
   <8-9>  MnSymbolE-Bold8
   <9-10> MnSymbolE-Bold9
  <10-12> MnSymbolE-Bold10
  <12->   MnSymbolE-Bold12
}{}
\let\llangle\@undefined
\let\rrangle\@undefined
\DeclareMathDelimiter{\llangle}{\mathopen}%
                     {MnLargeSymbols}{'164}{MnLargeSymbols}{'164}
\DeclareMathDelimiter{\rrangle}{\mathclose}%
                     {MnLargeSymbols}{'171}{MnLargeSymbols}{'171}
\makeatother

\title{The NLS limit for bosons in a quantum waveguide\thanks{This work was supported by the German Science Foundation within the GRK 1838 ``Spectral theory and dynamics of quantum systems''.}}
\author{Johannes von Keler and Stefan Teufel\\[2mm]\em
Mathematisches Institut, Universit\"at T\"ubingen, Germany.\\[1mm]\small stefan.teufel@uni-tuebingen.de}

\begin{document}

\selectlanguage{english}

\maketitle
\begin{abstract}
We consider a system of $N$ bosons confined to a thin waveguide, i.e.\ to a region of space within an $\epsi$-tube around a curve in $\R^3$. 
We show that when taking simultaneously the NLS limit $N\to \infty$ and 
the limit of strong confinement $\epsi\to 0$, the time-evolution of such a system starting  in a state close to a Bose-Einstein condensate is approximately captured by a non-linear Schr\"odinger equation in one dimension. The strength of the non-linearity in this Gross-Pitaevskii type equation depends on the shape of the cross-section of the waveguide,  while the ``bending'' and the ``twisting'' of the waveguide contribute potential terms. Our analysis is based on an approach to mean-field limits developed by Pickl \cite{Pic08}.

\end{abstract}

\section{Introduction}

We consider a system of $N$ identical weakly interacting bosons confined to a thin waveguide, i.e.\ to a region  $\tube_\epsi\subset \R^3$ contained in an $\epsi$-neighborhood of a curve $c:\R\to\R^3$. The Hamiltonian of the system is
\begin{equation}\label{hamilton1}
 H_{\tube_\epsi}(t)= \sum_{i=1}^N  \left(- \Delta_{z_i} + V(t,z_i) \right)+  \sum_{i \leq j}  \frac{a}{\mu^3} \,w  \left( \frac{ z_i-z_j}{\mu}\right)\,,
\end{equation}
where $z_j\in \R^3$ is the coordinate of the $j$th particle, $\Delta_{z_j}$ the Laplacian on $\tube_\epsi$ with Dirichlet boundary conditions, $V$ a possibly time-dependent external potential and $w$ a positive pair interaction potential. The coupling $a:=\epsi^2/N$ is chosen such that for $N$-particle states supported along a fixed part of the curve   the interaction energy per particle remains of order one for all $N\in\N$ and $\epsi>0$. 
For $\beta>0$ the effective range of the interaction $\mu:= \left(\epsi^2/N\right)^\beta$ goes to zero for $N\to \infty$ and $\epsi\to 0$ and $\mu^{-3} w(\cdot/\mu)$ converges to a point interaction. We consider in the following only $\beta\in (0,1/3)$, the so called mean-field regime  where  $a/\mu^3$ still goes to zero. For  recent papers containing   concise reviews of the mean-field and NLS limit for bose gases we refer to \cite{LewNamRou14,NamRouSei15}. For a detailed discussion of bose condensation in general and also the problem of dimensional reduction we refer to \cite{LieSeiSolYng05}.

Let us give a somewhat informal account of our result before we discuss the details.
Assume that the initial state  $\psi^{N,\epsi}\in L^2_+(\tube_\epsi^N) := \bigotimes^N_{\rm sym} L^2(\tube_\epsi)$ has a one-particle density matrix $\gamma_1$, i.e.\ the operator with kernel
\begin{equation}\label{redendef}
\gamma_1(z,z') := \int \psi^{N,\epsi}(z,z_2,\cdots ,z_N) \,\overline \psi^{N,\epsi}(z',z_2,\cdots ,z_N) \D z_2 \cdots \D z_N\,,
\end{equation}
that
is asymptotically close to a projection $p=|\ph\rangle\langle \ph|$ onto a single particle state $\ph= \Phi_0\chi\in L^2(\tube_\epsi)$, where $\Phi_0$ is the wavefunction along the curve and $\chi$ is the ``ground state'' in the confined direction.
Then we show that all $M$-particle density matrices $\gamma_M(t)$ of the solution $\psi^{N,\epsi}(t)$ of the Schr\"odinger equation
\begin{align*}
\im \tfrac{\D}{\D t}  \psi^{N,\epsi} (t)= H_{\tube_\epsi} (t) \,\psi^{N,\epsi} (t)
\end{align*}
are asymptotically close to $|\ph(t)\rangle\langle\ph(t)|^{\otimes M}$, where $
\ph(t) = \Phi(t)\chi$ with $\Phi(t)$ the solution of the one-dimensional non-linear Schr\"odinger equation
\begin{align}\label{equ:grosspqwgintro}
 \im \partial_t \Phi(t,x) =  \left(-\tfrac{\partial^2}{\partial x^2} + V_{\rm geom}(x) + V(t,x,0)+ b |\Phi(t,x)|^2\right) \,\Phi(t,x) \qquad \mbox{with } \;\Phi(0)=\Phi_0\,.
\end{align}
 The strength $b$ of the nonlinearity depends on the details of the asymptotic limit. We distinguish two regimes: In the case of  moderate confinement the width $\epsi$ of the waveguide shrinks slower than the range $\mu$ of the interaction and $b= \int_{\Omega_{\rm f}} |\chi(y)|^4 \, \D^2 y \cdot \int_{\R^3} w(r)\, \D^3 r$, where $\Omega_{\rm f}$ is the cross section of the waveguide and $\chi$ the ground state of the 2$d$-Dirichlet Laplacian on $\Omega_{\rm f}$. In the case of  strong confinement the width $\epsi$ of the waveguide shrinks faster than the range $\mu$ of the interaction and $b=0$. The geometric potential $V_{\rm geom}(x)$ depends on the geometry of the waveguide and is the sum of two parts. The curvature $\kappa(x)$ of the curve contributes a negative potential $-\kappa(x)^2/4$, while the twisting of the cross-section relative to   the curve contributes a positive potential. 
 Note that quasi one-dimensional Bose-Einstein condensates in non-trivial geometric structures have been realised experimentally \cite{GoeVogKet01,HeRy09} and that the transport and manipulation of condensates in waveguides is a highly promising experimental  technique, see e.g.\ the review \cite{FZ}.

The rigorous derivation of the non-linear Gross-Piteavskii equation   from the underlying linear many-body Schr\"odinger equation has been a very active topic in mathematics during the last decade, however, almost  exclusively without the confinement to a waveguide. Then  the  Gross-Piteavskii equation \eqref{equ:grosspqwg} is still an equation on $\R^3$.
The first rigorous and complete derivation in $\R^3$ is due to Erd\"os, Schlein and Yau \cite{ErdSchYau07}. Their proof is based on the BBGKY hierarchy, a system of coupled equations for all $M$-particle density matrices $\gamma_M(t)$, $M=1,\ldots,N$. Independently Adami, Golse and Teta solved the problem in one dimension \cite{AGT}. Shortly after, Pickl developed an alternative approach \cite{Pic08} that turned out very flexible concerning time-dependent external potentials \cite{Pic10a}, non-positive interactions \cite{Pic10b}, and singular interactions \cite{KnoPic09}. Yet another approach based on Bogoliubov transformations and coherent states on Fock space was developed for the most difficult case $\beta=1$ in \cite{BenOliSch12}. Recently also corrections to the mean-field dynamics were established in \cite{GrMa13,NaNa15}.
There are also several lecture notes reviewing  the different approaches to the NLS-limit, e.g.\ \cite{Sch08,Gol13,BenPorSch15,Rou15}.
For our purpose  the approach of Pickl turned out fruitful and our proof follows his general strategy and uses   his formalism. However, since   the NLS limit in a geometrically nontrivial waveguide required also crucial modifications,   our paper is fully self-contained.

Also the problem of deriving lower dimensional effective equations for strongly confined bose gases has been considered before. In 
\cite{abdmehschweis05} the authors start with the Gross-Pitaevskii equation in dimension $n+d$ confined to a $n$-dimensional plane by a strong harmonic potential and derive 
an effective NLS in dimension $n$. In \cite{MehRay15} the reduction of the Gross-Pitaevskii equation in dimension two to an $\epsi$-neighbourhood of a curve is considered. In both cases this corresponds to first taking the mean field limit and then the limit of strong confinement. However, we will see that the two limits do not commute and thus, that a direct derivation of the Gross-Pitaevskii equation in lower dimension from the $N$-particle Schr\"odinger evolution in higher dimension is of interest. This was done for a gas confined to a plane in $\R^3$ in \cite{CheHol13}, and for a gas confined to a straight line in \cite{CheHol14} using the BBGKY-approach of \cite{ErdSchYau07}.

\section{Main result}

In order to explain our result in full detail we need to start with the construction of the wave\-guide~$\tube_\epsi$.
Consider a 
smooth curve $c:\R \rightarrow \R^3$   parametrized by arc-length, i.e.\ with  $\|c'(x)\|_{\R^3}=1$. Along the curve we define a frame  by picking an   orthonormal frame $(\tau(0), e_1(0), e_2(0))$ at $c(0)$ with $\tau(0)=c'(0)$ tangent to the curve and then defining $(\tau(x), e_1(x), e_2(x))$ by parallel transport along the curve, i.e.\ by solving 
\begin{equation*}{\small
\label{eq:DGLBishop}
\begin{pmatrix} \tau' \\ e'_1 \\ e'_2 \end{pmatrix} = \begin{pmatrix} 0 & \kappa_1 & \kappa_2 \\ - \kappa_1 & 0 & 0 \\ -\kappa_2 & 0 & 0 \end{pmatrix} \begin{pmatrix} \tau \\ e_1 \\ e_2 \end{pmatrix}}
\end{equation*}
with the components of the mean curvature vector  $\kappa_j:\R\to\R$ ($j=1,2$) given by
\[
\kappa_j(x) := \langle \tau'(x), e_j(x)\rangle _{\R^3} = \langle c''(x),e_j(x)\rangle_{\R^3} \, .
\]
Let the cross-section $\Omega_{\rm f}\subset \R^2$ of the waveguide be open and bounded and let $\theta: \R\to\R$ be a smooth function   that defines the twisting of the cross-section relative to the parallel frame. In order to define the thin waveguide it is convenient to introduce the following maps separately. 
Denote the scaling map by
\[
D_\epsi :  \R^3 \to   \R^3\,,\quad r=(x,y) \mapsto (x,\epsi y)=:r^\epsi\,,
\]
the twisting map by
\[
T_\theta : \R^3 \to   \R^3\,,\quad (x,y)\mapsto (x,T_{\theta(x)}y)\,,\quad \mbox{where } \; T_{\theta(x)} = \begin{pmatrix}
\cos \theta(x) & -\sin\theta(x)\\\sin\theta(x)&\cos\theta(x)\end{pmatrix}\,,
\]
and the embedding map by
\begin{equation*}\label{def:f}
f : \R^3\to \R^3 \,, \quad r=(x,  y_1,y_2 )\mapsto f (r) = c(x) +  y_1 e_1(x)+   y_2 e_2(x)\,.
\end{equation*}
The waveguide is now defined by first scaling, then twisting and finally embedding the set $\Omega := \R\times \Omega_{\rm f}\subset \R^3$ into a neighbourhood of $c(\R)$. For $\epsi$ small enough,  the map
\[
f_\epsi: \Omega:= \R\times \Omega_{\rm f}\to \R^3\,, r \mapsto f_\epsi(r) := f\circ T_\theta\circ D_\epsi (r)
\]
is, by Assumption~{\bf A1},  a diffeomorphism onto its range 
\[
\tube_\epsi  := f_\epsi (\Omega ) \subset \R^3\,,
\]
which defines the region in space accessible to the particles, i.e.\ the waveguide. 
Now the evolution of an $N$-particle system in a waveguide is given by the Hamiltonian \eqref{hamilton1}, which acts on $L^2(\tube_\epsi)^{\otimes N} \cong L^2({\tube_\epsi}^N) $ with  Dirichlet boundary conditions.

However, for the formulation and the derivation of our result  it is more convenient to always work on 
 the fixed, $\epsi$-independent product-domain $\Omega = \R\times \Omega_{\rm f}$ instead of the tube $\tube_\epsi$. This is achieved by the natural  unitary transformation.   For $\epsi$ small enough the map $f_\epsi$ is a diffeomorphism and therefore  the map  
\[
U_\epsi :   L^2({\tube_\epsi} )\to L^2 (\Omega )\,,\quad \psi \mapsto (U_\epsi\psi)(r ) :=\sqrt{\det Df_\epsi(r)} \;\psi(f_\epsi(r)) =: \sqrt{\rho_\epsi(r)} \;\psi(f_\epsi(r)) 
\]
is unitary. Using $(U_\epsi)^{\otimes N}$ we can unitarily map the waveguide Hamiltonian $H_{\tube_\epsi}(t)$ in \eqref{hamilton1} to
\begin{eqnarray}\label{equ:hamqwg}
H(t)&:= &(U_\epsi)^{\otimes N} H_{\tube_\epsi}(t)(U_\epsi^*)^{\otimes N} + \sum_{i=1}^N \tfrac{1}{\epsi^2} V^\perp(y_i)\\& =&   \sum_{i=1}^N \left( -\left(  U_\epsi \Delta U_\epsi^*  \right)_{z_i}  + \tfrac{1}{\epsi^2}V^\perp(y_i)+V(t, f_\epsi(r_i))\right)+ a \sum_{i \leq j}\frac{1}{\mu^3} \,w  \left( \frac{ f_\epsi(r_i)-f_\epsi(r_j)}{\mu}\right)\,,\nonumber
\end{eqnarray}
  where we allow for an additional confining potential $V^\perp:\Omega_{\rm f} \to \R$.
  We denote the lowest eigenvalue of $-\Delta_y + V^\perp(y)$ on $\Omega_{\rm f}$ with Dirichlet boundary conditions by $E_0$ and the corresponding real valued and normalised eigenfunction by $\chi$.
  
  We will consider simultaneously the mean-field limit $N\to \infty$ and the limit of strong confinement $\epsi\to 0$ for the time-dependent Schr\"odinger equation 
  with Hamiltonian $H(t)$ on the Dirichlet domain $ D(H(t)) \equiv D(H)=H^2(\Omega^N) \cap H^1_0(\Omega^N)$.
  Recall   that the effective coupling $a$ is given by $a=\epsi^2/{N}$  and the effective range of the interaction by  $\mu = (\epsi^{2}/N)^{\beta}$. 
  
 Compared to the standard $N$-particle Schr\"odinger operator we thus have in \eqref{equ:hamqwg} the shrinking domain and the strongly confining potential $V^\perp$, a  pair interaction that is no longer exactly a function of the separation $r_i-r_j$ of two particles, and a modified kinetic energy operator.

  \begin{lem} \label{LapKoord} The Laplacian in the new coordinates has the form
  \[
   U_\epsi \Delta U_\epsi^*  = - \left(\partial_x +  \theta' (x)L \right)^2\;-\; \tfrac{1}{\epsi^2} \Delta_y\; - \;V_{\rm bend}(r) \;-\; \epsi \,S^\epsi\,,
   \]
   where
   \begin{eqnarray*}
  L &=&    y_1\partial_{y_2} - y_2\partial_{y_1}
      \,,\\[2mm]
  V_{\rm bend}(r)&=&  - \frac{\kappa(x)^2}{4\rho_\epsi(r)^2} - \epsi\,\frac{T_{\theta(x)}y\cdot\kappa(x)''}{2\rho_\epsi(r)^3} -
  \epsi^2\, \frac{5(  T_{\theta(x)}y\cdot\kappa'(x))^2}{4\rho_\epsi(r)^4}\,,
  \\[2mm]
   S^\epsi  &=&  \left(\partial_x  +\theta'  (x) L \right)  s^\epsi(r)  \left(\partial_x  + \theta' (x)  L \right)\,,\\[2mm]
   \rho_\epsi(r) &=& 1- \epsi\, T_{\theta(x)}y\cdot \kappa(x)\,,\quad\mbox{and} \quad s^\epsi(r) = \frac{\rho_\epsi^2(r)-1}{\epsi\,\rho_\epsi^2(r)} 
   \,.
   \end{eqnarray*}
  \end{lem}
  \begin{proof} This is an elementary computation and the result is, somewhat implicitly, used in many papers on quantum waveguides, see e.g.\  \cite{Kre08} and references therein.
 The explicit result using our notation is 
derived  in the introduction of \cite{HaLaTe14} for the case $\theta\equiv 0$ and yields the corresponding expression with $\partial_x$ instead of $\partial_x+\theta L$. Now 
  the rotation by the angle $\theta(x)$ in the $y$-plane is implemented on $L^2(\R^3)$ by the operator $R(\theta(x))=\E^{\theta(x)  ( y_1\partial_{y_2} - y_2\partial_{y_1} )}$, such that 
  \[
  R(\theta(x))^* \,\partial_x\,R(\theta(x))= \partial_x + \theta'   L\,.
  \]
    \end{proof} 
Before stating our main result we give a list of assumptions. 
\begin{itemize}
\item[\bf A1] {\em Waveguide}: Let $\Omega_{\rm f}\subset\R^2$ be open and bounded. Let $c:\R\to \R^3$  be injective and  six times continuously differentiable  with all derivatives bounded, i.e.\ $c \in C^6_{\rm b}(\R,\R^3)$,
 and such that $\|c'(x)\|_{\R^3}\equiv 1$. To avoid overlap of different parts of the waveguide injectivity is not sufficient and we assume  that there are constants $c_1,c_2>0$ such that
\[
\|c(x_1) - c(x_2)\|_{\R^3} \geq \min\{ c_1|x_1-x_2|, c_2\}\,.
\]
Finally let $\theta:\R\to\R$  satisfy  $\theta\in C^5_{\rm b}(\R)$.
  \item[\bf A2]  {\em Interaction}: Let the interaction potential  $w$ be a non-negative, radially symmetric function such that 
  $w(r) = \tilde w(|r|^2)$ for a function $\tilde w\in C^2(\R)$ with support in $(-1,1)$.\\[1mm]
If the waveguide is straight and untwisted, i.e.\ if $f= T_\theta= {\rm id}$,  then we only assume that
   $w$ is a non-negative  function in $L^2(\R^3; \D^3 r)\cap L^1(\R^3; (1+|r|)\, \D^3 r)$.

  \item[\bf A3] {\em External potentials}: Let the external single particle potential $V:\R\times\R^3\to \R$ for each fixed $t\in\R$ be bounded and four times  continuously differentiable with bounded derivatives,  $ V(t,\cdot) \in C^4_{\rm b}(\R^3)$. Moreover assume that the map  $\R\to L^\infty(\Omega)$, $t\mapsto V(t,\cdot)$ is differentiable and that $\dot V(t,\cdot)\in C_{\rm b}^1(\R^3)$.
  
Let the confining potential $V^\perp:\Omega_{\rm f}\to \R$ be relatively bounded with respect to the Dirichlet Laplacian on $\Omega_{\rm f}$ with relative bound smaller than one.
  \end{itemize}

\begin{rem}
\begin{enumerate}[(a)]
\item
Note that for geometrically non-trivial waveguides we will have  to Taylor expand the interaction~$w$ up to second order, hence condition {\bf A2}. Otherwise the much weaker condition formulated for straight and untwisted waveguides suffices. 
Note also that any radially symmetric function can be written uniquely in the form $w(r) = \tilde w(|r|^2)$ and that  the regularity we need for the Taylor expansion is most conveniently formulated in terms of $\tilde w$.
\item 
The high regularity requirements for the wave\-guide in {\bf A1} and the external potential in {\bf A3} are only needed  to ensure the existence of global solutions of the NLS equation \eqref{equ:grosspqwgintro} that remain bounded in~$H^2(\R)$.  
\end{enumerate}
\end{rem}

Let $\psi^{N,\epsi}(t)$ be the solution to the time-dependent $N$-particle Schr\"odinger equation 
\begin{align}\label{equ:schrodinger}
\im \tfrac{\D}{\D t}  \psi^{N,\epsi} (t)= H (t) \,\psi^{N,\epsi} (t)
\end{align}
with the Hamiltonian $H(t)$ defined  in \eqref{equ:hamqwg} and $\psi^{N,\epsi} (0)\in D(H(t)) \equiv H^2(\Omega^N) \cap H^1_0(\Omega^N)$. 

In order to study simultaneously the mean-field limit $N\to \infty$ and the limit of strong confinement $\epsi\to 0$, we consider  families of initial data $\psi^{N,\epsi} (0)$ along sequences $(N_n,\epsi_n)\to (\infty,0)$.

\begin{defn}
 For $\beta\in \left(0,\frac13\right)$ we call a sequence $(N_n,\epsi_n)$ in $\N\times (0,1]$ {\bf admissible}, if      
  \begin{equation}\label{admis}
   \lim_{n\to \infty} (N_n,\epsi_n) = (\infty,0)
  \qquad\mbox{and}\qquad    \lim_{n\to \infty}   \frac{(\epsi_n)^{\frac43}}{\mu_n} =0  \qquad \mbox{for}\qquad \mu_n := \left( \frac{\epsi_n^2}{N_n}\right)^\beta\,.
  \end{equation}
 We say that the sequence $(N_n,\epsi_n)$ is {\bf moderately confining}, if, in addition, 
  \[
  \lim_{n\to \infty}  \frac{\mu_n}{\epsi_n} = 0 \,,
  \]
  i.e.\ if the effective range $\mu$ of the interaction shrinks faster than the width $\epsi$ of the waveguide.
We say that the sequence $(N_n,\epsi_n)$ is {\bf strongly confining}, if instead 
  \[
  \lim_{n\to \infty}  \frac{\epsi_n}{\mu_n} = 0 \,,
  \]
i.e.\ if the width of the waveguide is small even on the scale of the interaction.
  \end{defn}

Note that the admissibility condition  in \eqref{admis} requires that the width $\epsi$ of the waveguide    
cannot shrink too slowly compared to the range of the interaction $\mu$. This is a technical requirement that simplifies the proof considerably. It assures that the energy gap between $E_0$ and the first excited state in the normal direction, which is of order $\frac{1}{\epsi^2}$, grows sufficiently quickly so
 that transitions into excited states in the normal direction become negligible at leading order. 
In the following we will be   concerned  almost exclusively with the case of moderate confinement, where the effective one dimensional equation is nonlinear. The analysis of the strongly confining case turns out to be  much simpler. 

Before we can formulate our  precise assumptions on the family of initial states, we need to introduce  the one-particle energies. 
For $\psi\in D(H)$   the ``renormalised energy per particle'' is
\begin{align*}
  E^{\psi} (t):= \tfrac{1}{N}\big\langle \psi ,H(t)\,\psi  \big\rangle_{L^2(\Omega^{N})} -    \tfrac{E_0}{\epsi^2} \,,
 \end{align*}
and for $\Phi\in H^2(\R)$ let the ``effective energy per particle'' be
\begin{align}\label{equ:enggross2}
  E^{\Phi} (t):&= \Big\langle \Phi ,\underbrace{\left(-\tfrac{\partial^2}{\partial x^2} - \tfrac{\kappa(x)^2}{4} +   |\theta'(x)|^2 \,\|L\chi\|^2   + V(t,x,0)+  \tfrac{b}{2} |\Phi |^2 \right)}_{\displaystyle =: \mathcal{E}^\Phi (t)}  \Phi  \Big\rangle_{L^2(\R)}\,.
\end{align}
Recall that $\chi$ is the ground state wave function of $-\Delta_y + V^\perp(y)$ on $\Omega_{\rm f}$ with Dirichlet boundary conditions and $E_0$ the corresponding ground state eigenvalue. As with $L^2_+(\tube_\epsi)$, we also denote   the symmetric subspace of $L^2(\Omega^N)$ by  $ L^2_+(\Omega^N):=\bigotimes^N_{\rm sym} L^2(\Omega)$.

\begin{itemize}

  \item[\bf A4] {\em Initial data}:
   Let the family of initial data   $\psi^{N,\epsi}(0)\in D(H)\cap L^2_+(\Omega^N)$, $\|\psi^{N,\epsi}(0)\|^2=1$,   be close to a condensate with single particle wave function $\varphi_0 = \Phi_0\chi$ for some $\Phi_0\in H^2(\R)$   in the following sense: for some admissible sequence $(N,\epsi)\to (\infty,0)$   it holds that 
  \begin{align*}
   \lim_{(N,\epsi)\to (\infty,0)} \Tr_{L^2(\Omega)}  \big |\gamma^{N,\epsi}(0)-|\varphi_0 \rangle \langle \varphi_0|  \big |=0\,,
  \end{align*}
  where $\gamma^{N,\epsi}(0)$ is the one particle density matrix of $\psi^{N,\epsi}(0)$, cf.\ \eqref{redendef}. In addition we assume that also  the energy per particle converges,
  \begin{align*}
   \lim_{(N,\epsi)\to (\infty,0) }|E^{\psi^{N,\epsi}(0)} (0)-E^{\Phi_0} (0)|=0\,.
  \end{align*}

\end{itemize}

Finally,  let $\Phi(t)$ be the corresponding solution of the effective nonlinear Schr\"odinger equation
\begin{align}\label{equ:grosspqwg}
 \im \partial_t \Phi(t) =  \underbrace{\left(-\tfrac{\partial^2}{\partial x^2} - \tfrac{\kappa(x)^2}{4} + |\theta'(x)|^2 \,\|L\chi\|^2 + V(t,x,0)+ b |\Phi(t)|^2\right)}_{\displaystyle =: \,h^\Phi(t)} \,\Phi(t) \qquad \mbox{with } \;\Phi(0)=\Phi_0, 
\end{align}
 where   
\begin{equation*}\label{bdef}
b:= \left\{ \begin{array}{cl}
 \int_{\Omega_{\rm f}} |\chi(y)|^4 \, \D^2 y \cdot \int_{\R^3} w(r)\, \D^3 r &
 \mbox{ in the case of moderate confinement,}\\[1mm]
 0&  \mbox{ in the case of strong confinement.}
\end{array}\right.
\end{equation*}

The unique existence and properties of solutions to \eqref{equ:schrodinger} and \eqref{equ:grosspqwg} are well known  and briefly discussed in Appendix\,\ref{app:regsol}.

\begin{thm}\label{thm:thm1} Let the waveguide satisfy assumption {\rm \bf A1} and
let the potentials satisfy assumptions {\rm \bf A2} and {\rm \bf A3}. For  $\beta\in \left(0,\frac13\right)$ let $\psi^{N,\epsi}(0)$ be a family of initial data satisfying {\rm \bf A4}.
Let $\psi^{N,\epsi}(t)$ be the solution of the $N$-particle Schr\"odinger equation \eqref{equ:schrodinger} with initial datum $\psi^{N,\epsi}(0)$ and $\gamma^{N,\epsi}_M(t)$ its $M$-particle reduced density matrix. Let 
 $\Phi(t)$ be the solution of the effective  equation \eqref{equ:grosspqwg} with initial datum   $\Phi_0$. 
Then   for any $t\in \R$  and any $M\in\N$
\[
   \lim_{(N,\epsi)\to (\infty,0)} \Tr \Big |\gamma^{N,\epsi}_M(t)-|\Phi(t) \chi  \rangle \langle \Phi(t) \chi |^{\otimes M}  \Big |=0\,,
\]
and
\[
 \lim_{(N,\epsi)\to (\infty,0)} \left|E^{\psi^{N,\epsi}(t)}(t)-E^{\Phi(t)}(t)\right|=0
\]
where the limits are taken along the sequence  from {\rm \bf A4}. 
 \end{thm}

\begin{rem}
 \begin{enumerate}[(a)]
     
    \item In Assumption {\bf A4} we assume that the initial state is close to a complete Bose-Einstein condensate. 
     To show that the ground state  of a bose gas is actually of this form is in itself an important and difficult problem. For a straight wave guide and the case $\beta=1$ this was shown in \cite{LieSeiYng03}, see also   \cite{LieSeiSolYng05} for a detailed review and \cite{SchYng06}.           The analysis of ground states in geometrically non-trivial wave guides is, as far as we know, an open problem.
   For  the latest  results for  $\beta \in (0,1) $, but without strong confinement, we refer to \cite{LewNamRou14}.

 \item The assumption in {\bf A2} that the interaction   $w$ is non-negative seems to be crucial to our proof, although it is used only once in the proof of the energy estimate of Lemma~\ref{lem:energyestimate}. The results of \cite{CheHol14} suggest, however, that also our result should hold for interactions with a certain negative part.
 
    \item The negative part $-\kappa(x)^2/4$ of the geometric potential  stemming from the curvature $\kappa(x)$ of the curve is often called the bending potential, while the positive part $|\theta'(x)|^2\|L\chi\|^2$ is called the twisting potential. Both appear in exactly the same form also for non-interacting particles in a waveguide, as they originate just from the transformation of the Laplacian in Lemma~\ref{LapKoord}. See also \cite{Kre08} for a review in the one-particle case.
    
\item
 One could also consider a waveguide with a cross-section that varies along the curve, e.g.\ having constrictions or thickenings. But then $E_0=E_0(x)$ would be a function of $x$ and an effective potential of size $\frac{E_0(x)}{\epsi^2}$ would appear in the effective equation. As a consequence also the kinetic energy in the $x$-direction would be of order $\frac{1}{\epsi^2}$, i.e.\ $\|\Phi\|_{H_1(\R)}^2 = \Or(\frac{1}{\epsi^2})$. It is conceivable that a similar result to Theorem~\ref{thm:thm1} holds also in this setting of large tangential energies. However, this is a much more difficult problem, since transitions into excited normal modes will be energetically possible. Using adiabatic theory, this problem is treated in the single-particle case in   \cite{WT,LT,HaLaTe14}.
  \item
 Another interesting modification of the setup is the confinement only by potentials, without the Dirichlet boundary. Also this would  introduce additional technical complications, since in this case the map $f$ is no longer a global diffeomorphism and has to be cut off, c.f.\ \cite{WT}.  
 \item Let us breifly comment on the main differences of our result  compared to the   work  of Chen and Holmer \cite{CheHol14}. While our focus is on geometrically non-trivial wave guides, the  authors of \cite{CheHol14} consider   confinement by a harmonic potential of constant shape to a straight line.   However, their main focus are attractive pair interactions, more precisely pair potentials with $\int_{\R^3} w(r)\D r \leq 0$,   a situation which is excluded in our result. On the other hand, at least in the case of a straight wave guide, our approach needs much less regularity for $w$ and can incorporate external time-dependent potentials. Finally, our proof yields also convergence rates, which, as far as we understand, is not the case for \cite{CheHol14}. As explained below, we refrain from stating these rates because they are quite complicated  and most likely far from optimal.
 
  \item In \cite{LieSeiYng03} the authors exhibit five different scaling regimes with different effective energy functionals for the ground state energy. Note that a direct comparison with our two regimes is not sensible for two reasons:
   First we assume $\beta\in(0,\frac13)$ while in \cite{LieSeiYng03} the Gross-Pitaevskii scaling $\beta=1$ is considered. As a consequence, in \cite{LieSeiYng03} the scattering length, i.e.\ the range of the interaction $w$, is always small compared to the small diameter $\epsi$ of the wave guide. The siutation $\epsi/\mu\to 0$ (what we called strong confinement) does not occur for $\beta =1$.
 Secondly, \cite{LieSeiYng03} is specifically concerned with the ground state energy, where some terms in the energy functional can become negligible or can take a specific form depending on details of the ground state.

    \end{enumerate}

\end{rem}

\subsection*{Acknowledgements}
We thank  Steffen Gilg, Stefan Haag, Christian Hainzl, Jonas Lampart, S\"oren Petrat, Peter Pickl,   Guido Schneider, and Christof Sparber  for helpful discussions. The support by the German Science Foundation (DFG) within the GRK 1838 ``Spectral theory and dynamics of quantum systems'' is gratefully acknowledged.
\subsection*{Ethical Statement}
Funding: This work was funded by the German Science Foundation (DFG) within the GRK 1838. Conflict of Interest: The authors declare that they have no conflict of interest.

\section{Structure of the proof and the main argument}

In the proof we will not directly control the difference $\Tr \big |\gamma^{N,\epsi}_M(t)-|\varphi(t)  \rangle \langle \varphi(t) |^{\otimes M}  \big |$, but use a 
 functional $\alpha(\psi^{N,\epsi}(t),\varphi(t))$ introduced by Pickl \cite{Pic08,KnoPic09, Pic11} to  measure the ``distance'' between    $\psi^{N,\epsi} $ and $\varphi$. For this measure of distance our proof yields also rates of convergence, which could be translated into rates of convergence also for $\Tr \big |\gamma^{N,\epsi}_M(t)-|\varphi(t)  \rangle \langle \varphi(t) |^{\otimes M}  \big |$.
 However, since these rates are presumably far from optimal, we refrain from stating them explicitly.
 
  The functional $\alpha$ is constructed from the following projections in the $N$-particle Hilbert space.
 
\begin{defn}\label{def:pP}
 Let $p$ be an orthogonal projection in the one-particle space $L^2(\Omega)$.
 
For  $i \in \{1, \dots ,N\} $   define on $L^2(\Omega )^{\otimes N}$ the projection operators
    \begin{align*}
      p_i:= \underbrace{ \id \otimes \cdots \otimes \id}_{i-1 \; \mathrm{times}}  \otimes \, p \otimes \underbrace{ \id \otimes \cdots  \otimes \id}_{N-i \; \mathrm{times}}
   \qquad\mbox{ and } \qquad
      q_i:=\id - p_i\,.
    \end{align*}

   For   $0 \leq k \leq N$ let
   \begin{align*}
   P_{k} :=  \Big( q_1\cdots q_k p_{k+1} \cdots p_N \Big)_{\mrm{sym}}: =  
   \sum_{\substack{J\subset\{1,\ldots, N\}\\ |J|=k}} \,\prod_{j\in J} q_j \prod_{j\notin J} p_j\,.
    \end{align*}
  For $k<0 $ and $k>N$ we set $P_{k } =0$.

\end{defn}

We will use the many-body projections $P_k$ exclusively for $p  =   |\varphi\rangle\langle\varphi|$, the orthogonal projection onto the subspace spanned by the condensate state   $\varphi \in L^2(\Omega) $ with $\norm{\varphi}_{L^2(\Omega)}=1$. However,  a number of simple algebraic  relations, like
  \begin{equation}\label{Prel}
   \sum_{k=0}^N P_{k } = \id\,, \qquad
    \sum_{i=1}^N q_i P_{k } = k P_{k } \,,
  \end{equation}
  hold independently of the special choice for $p$ and will turn out very useful in the analysis of the mean field limit.
     The first identity in \eqref{Prel} follows from the fact that $q_i+p_i=\id$.
  For the second identity note that together with the first identity we have
    \[
     \sum_{i=1}^N q_i =  \sum_{i=1}^N q_i \sum_{k'=0}^N P_{k' } = \sum_{k'=0}^N  \sum_{j=1}^N  q_i P_{k' }=  \sum_{k'=0}^N  k'  P_{k' }\,.
    \]
    Projecting with $P_k$ yields the second identity, since $P_kP_{k'} = \delta_{k,k'} P_k$.

 \begin{defn}\label{hutdef}
For any function $f:\N_0\to \R$  define the bounded linear operator
\[
\widehat f :L^2(\Omega^N)\to L^2(\Omega^N)\,,\quad \psi\mapsto \widehat f \psi:=
 \sum_{k=0}^N f(k)  P_{k }   \psi 
\]
and the functional $\alpha_f: L^2(\Omega^N)\times L^2(\Omega)\to \R $   
\[
 \alpha_{f}\big(\psi,\varphi\big):=    \left\langle \psi ,\widehat f\,   \psi  \right\rangle_{L^2(\Omega^{N})}  = \sum_{k=0}^N   f(k) \,\left\langle \psi ,  P_{k }\,   \psi  \right\rangle_{L^2(\Omega^{N})} \,.
\]
\end{defn}

The heuristic idea behind this definition is the following. The operator  $P_k$ projects onto the subspace of $L^2(\Omega^N)$ of those states, where exactly $k$ out of the $N$ particles are not     condensed into $\ph$.
Components of $\psi\in L^2(\Omega^N)$  with $k$ particles outside the condensate are weighted by $f(k)$ in   $\alpha_f(\psi, \ph)$. In order to obtain a useful measure of distance between $\psi$ and the condensate $\ph^{\otimes N}$, the function $f$ should thus be   increasing   and $f(0)$ should be (close to) zero.
For  $n(k) := \sqrt{k/N}$ it is easily seen  that the   functional $ \alpha_{n^2} $ is   a good measure for condensation:  Using the shorthand 
\[
  \llangle \cdot,\cdot \rrangle := \langle \cdot , \cdot \rangle_{L^2(\Omega^N)}\,,
\]
we find for any symmetric $\psi\in L^2(\Omega^N)$ 
\begin{eqnarray}
\alpha_{n^2}(\psi,\varphi)  &= &\sum_{k=0}^N  \frac{k}{N}  \left\llangle \psi ,P_{k }   \psi  \right\rrangle   \;\stackrel{\eqref{Prel}}{=}\;  \sum_{k=0}^N \sum_{i=1}^N  \frac{1}{N}  \left\llangle \psi ,q_iP_{k }   \psi  \right\rrangle\nonumber\\& \stackrel{\rm symmetry}{=}&
\sum_{k=0}^N    \left\llangle \psi ,q_1P_{k }   \psi  \right\rrangle \;\stackrel{\eqref{Prel}}{=}\;
\llangle  \psi , q_1   \psi  \rrangle 
=\|q_1\psi\|^2\,.\label{n2comp}
\end{eqnarray}
And in general we have the  following equivalences.
\begin{lem}\label{lem:equi}
 Let $\psi^N \in L^2_+(\Omega^N)$ be a sequence of normalised $N$-particle wave functions and let $\gamma^M_N$ be the sequence of corresponding $M$-particle density matrices for some fixed $M\in\N$. Let   $\varphi \in L^2(\Omega)$ be normalised. Then the following statements are equivalent:
\begin{enumerate}[(i)]
 \item  $\lim_{N\rightarrow \infty}  \alpha_{n^a}(\psi^N,\varphi)=0 $  for some $ a>0 $
 \item  $\lim_{N\rightarrow \infty}  \alpha_{n^a}(\psi^N,\varphi)=0 $  for any $ a>0 $

    \item     $ \lim_{N\rightarrow \infty} \left\| \gamma^N_M - | \varphi  \rangle \langle\varphi |^{\otimes M}\right\| = 0$    for all $M\in \N$    
\item   $ \lim_{N\rightarrow \infty} {\rm Tr} \left| \gamma^N_M - | \varphi  \rangle \langle\varphi |^{\otimes M}\right|=0$  for all $M\in \N$    

\item $ \lim_{N\rightarrow \infty} {\rm Tr} \left| \gamma^N_1 - | \varphi  \rangle \langle\varphi | \right|=0$ 

\end{enumerate}
\end{lem}
The proof of this lemma collects different statements somewhat scattered in the literature, c.f.\ \cite{Pic11, KnoPic09}. Since the claim is at the basis of our result and since the proof is short and simple, we give it at the end of  Subsection~\ref{sec:beta} for the convenience of the reader.

In the proof of our main theorem  we will work with the   functional  $\alpha_{m}$, where 
\begin{equation}\label{mdef}
 m (k):= \begin{cases}
         n(k) & \mathrm{for}\; k \geq N^{1-2\xi}\\
	\frac{1}{2}(N^{-1+\xi}k+N^{-\xi})& \mathrm{else} 
        \end{cases}
\end{equation}
for some  $ 0<\xi<\frac12$ to be specified below.
 Since $n(k) \leq m (k) \leq \max(n(k),N^{-\xi})$ holds for all $ k \in \N_0$,  convergence of $\alpha_{m }$ to zero is equivalent to convergence of $ \alpha_{n}$ to zero    and thus to all cases in Lemma~\ref{lem:equi}.
We will   use the shorthand
\begin{align*}
 \alpha_{m }(t):= \alpha_{m } \big(\psi^{N,\epsi}(t), \varphi(t)\big)  
\end{align*}
when we evaluate the functional $\alpha_{m }$ on the solutions to the time-dependent equations.
Finally, the quantity that we can actually control in the proof  is
\begin{equation}\label{alphadef}
 \alpha_\xi (t):= \alpha_{m }(t)+ \left|E^{\psi^{N,\epsi}(t)}(t)-E^{\varphi(t)}(t)\right|\,.
\end{equation}

We will now state two key propositions and then give the proof of Theorem~\ref{thm:thm1}. The simple strategy is to show bounds for the time-derivative of  $\alpha_\xi$ and then use  Gr\"onwall's inequality. 
With the expression from Lemma~\ref{LapKoord} for the Laplacian in the adapted coordinates we find that
\begin{eqnarray*}
H(t) &=&  \sum_{i=1}^N \left( -\left(  U_\epsi  \Delta U_{\epsi}^* \right)_{r_i}  +\frac{1}{\epsi^2}V^\perp(y_i)+ V(t, f_\epsi(r_i))\right)+ a \sum_{i < j}\frac{1}{\mu^3} \,w  \left( \frac{ f_\epsi(r_i)-f_\epsi(r_j)}{\mu}\right)\\
&=&  \sum_{i=1}^N \left( -\frac{\partial^2}{\partial x_i^2} -(\theta'(x_i)   L_i) ^2 - \frac{1}{\epsi^2} \Delta_{y_i} +\frac{1}{\epsi^2}V^\perp(y_i) + V(t, r_i^\epsi) - \frac{\kappa(x_i)^2}{4} + R^{(1)}_i\right)
\\
&&\qquad + \; \frac{1}{N-1} \sum_{i<j} w^{\epsi, \beta,N}_{ij}  
\end{eqnarray*}
with
\[
R^{(1)}_i := - \partial_{x_i}  \theta'(x_i)   L_i - \theta'(x_i)   L_i \partial_{x_i}  \;+ \;\left(V_{\rm bend}(r_i)+\frac{\kappa(x_i)^2}{4}
\right) \;-\; \epsi \,S^\epsi _i
\]
and
\[
w^{\epsi, \beta,N}_{ij}(r_1,\ldots,r_N) :=  (N-1)\frac{a}{\mu^3} \,w  \left( \frac{ f_\epsi(r_i)-f_\epsi(r_j)}{\mu}\right)  \,.
\]

\begin{prop}\label{lem:beta.g} Let the assumptions of Theorem~\ref{thm:thm1}  hold and let $\alpha_\xi(t)$ be given by \eqref{alphadef}. Then the time-derivative of $\alpha_\xi(t)$ is bounded by
 \begin{align*}
  \left|\frac{\D}{\D t} \,\alpha_\xi(t)\right| \leq    2|\mathrm{I}(t)| +  |\mathrm{II}(t)| +  2|\mathrm{III}(t)|+ |\mathrm{IV}(t)|
 \end{align*}
with
\begin{eqnarray*}
 \mathrm{I}(t)&:=&N\left\llangle \psi^{N,\epsi}(t), p_1 p_2  \left[   w^{\epsi, \beta,N}_{12}-  b|\Phi(t,x_2)|^2 ,\widehat m \right] p_1q_2 \psi^{N,\epsi}(t) \right\rrangle\\
\mathrm{II}(t)&:=& N\left\llangle \psi^{N,\epsi}(t), p_1 p_2  \left[    w^{\epsi, \beta,N}_{12}, \widehat m\right] q_1q_2 \psi^{N,\epsi}(t) \right\rrangle\\ 
\mathrm{III}(t)&:= &N\left\llangle \psi^{N,\epsi}(t), p_1 q_2  \left[    w^{\epsi, \beta,N}_{12}-   b|\Phi(t,x_1)|^2 , \widehat m\right] q_1q_2 \psi^{N,\epsi}(t) \right\rrangle\\
\mathrm{IV}(t)&:=& \left|\left\llangle \psi^{N,\epsi}(t), \dot V(t,x_1,\epsi y_1) \psi^{N,\epsi}(t) \right\rrangle - \langle \Phi, \dot V(t,x_1,0) \Phi \rangle_{L^2(\R)}\right|  \\
&&+\;    2N \left\llangle \psi^{N,\epsi}(t), p_1    \left[   V(t,x_1,\epsi y_1)-V(t,x_1,0) , \widehat m\right] q_1 \psi^{N,\epsi}(t) \right\rrangle\\
&&+\;     2N \left\llangle \psi^{N,\epsi}(t), p_1    \left[    (\theta'(x_1)L_1)^2 + |\theta'(x)|^2 \,\|L\chi\|^2 ,\widehat m \right] q_1 \psi^{N,\epsi}(t) \right\rrangle \\
&&+\;       2N \left\llangle \psi^{N,\epsi}(t), p_1    \left[  R^{(1)}_1 ,\widehat m \right] q_1 \psi^{N,\epsi}(t) \right\rrangle
\end{eqnarray*}
\end{prop}
The three terms I--III contain the two-body interaction and are delicate to bound because of the factor $N$ in front.
Very roughly speaking, Term I is small because in between the projection $p_1$ onto the state $\ph$ in the first variable the full interaction and the mean-field interaction cancel each other at leading order.  In Term~II and Term~III the full interaction $w^{\epsi, \beta,N}_{12}$ acting on the range of $q_1q_2$ becomes singular as $(N,\epsi)\to(\infty,0)$, but both   can still be bounded in terms of $\alpha_\xi$, however, with considerable effort. The one-particle contributions in Term~IV are rather easy to handle, as all potentials appearing remain bounded also on the range of $q_1$. However, the first line of IV is only small if $\psi$ is close to the condensate.
 
 In the following estimates we use 
the function $g(t)>0$   given in terms of  the a priori bound on the energy per particle, 
\[
|E^{\psi^\epsi_N(t)}  (t) | \;\leq \;|E^{\psi^\epsi_N(0)}  (0) | +   \int_0^t \| \dot V(s,\cdot)\|_{L^\infty(\Omega)}\D s \;=:\; g^2(t)-1\,.
\]
  If the external potential is time-independent, then $g^2(t)\equiv1+|E^{\psi^\epsi_N(0)}  (0) |$.

We defer  the proof of   Proposition~\ref{lem:beta.g} and also of the following one to Section~4.
 \begin{prop}\label{lem:3termeg} For moderate confinement we have the bounds
 \[
\begin{array}{ll}
 |\mathrm{I}(t)|\lesssim g(t)\,\left\| \Phi(t)\right\|^3_{H^2(\R)} \left( \frac{\mu}{\epsi} + N^\xi \epsi + \frac{ \epsi^2}{\mu^{\frac32}}\right) \,,&
 |\mathrm{II}(t)|\lesssim   \norm{\Phi(t)}_{H^2(\R)}^2  {\alpha_\xi(t)} +   N^\xi  \frac{a}{\mu^3} \,,
 \\
  |\mathrm{III}(t)|  \lesssim  g(t)\,\left\| \Phi(t)\right\|^\frac{3}{2}_{H^2(\R)} \left( \alpha_\xi(t) + \frac{\mu}{\epsi} +\frac{a}{\mu^3}+  \frac{ \epsi^2}{\mu^{\frac32}} \right) 
\,,&
 |\mathrm{IV}(t)| \lesssim   \alpha_\xi(t)+\epsi\,\|\Phi(t)\|_{H^2(\R)} + g(t) N^\xi \epsi\,.
 \end{array}
\]
For strong confinement we have the bounds
 \[
\begin{array}{ll}
 |\mathrm{I}(t)|\lesssim \mu \left\| \Phi(t)\right\|^2_{H^2(\R)}  \,,&
 |\mathrm{II}(t)|\lesssim   (\alpha_\xi(t)+\mu)\, \left\| \Phi(t)\right\|_{H^2(\R)}  \,,
 \\[2mm]
  |\mathrm{III}(t)|  \lesssim  (\alpha_\xi(t)+\mu)\, \left\| \Phi(t)\right\|_{H^2(\R)} \,,&
 |\mathrm{IV}(t)| \lesssim   \alpha_\xi(t)+\epsi\,\|\Phi(t)\|_{H^2(\R)} + g(t) N^\xi \epsi\,.
 \end{array}
\]
\end{prop}

Here and in the remainder of the paper we use the notation
$A \lesssim B$ to indicate that there exists a constant $C\in\R$ independent of all ``variable quantities'' $\epsi,N,t,\xi, \Psi^{\epsi,N}(0)$, and $ \Phi_0$ such that $A\leq CB$. Note that $C$ can depend on ``fixed quantities'' like the shape of the waveguide determined by $c,\theta, \Omega_{\rm f}$, and also on the  potentials $V,w$, $V^\perp$ and on $\beta$.

\begin{proof}[Proof of Theorem\,\ref{thm:thm1}]
Combining Propositions~\ref{lem:beta.g} and\,\ref{lem:3termeg} we obtain for the case of moderate confinement  that
\[
 \left|\frac{\D}{\D t} \,\alpha_\xi(t)\right| \leq C\, g(t) \,\left\| \Phi(t)\right\|^3_{H^2(\R)}\left( \alpha_\xi(t) + \frac{\mu}{\epsi} +  \frac{ \epsi^2}{\mu^{\frac32}} +N^\xi \epsi + N^\xi \frac{a}{\mu^3}\right)
 \]
for a constant $C<\infty$ independent of $t,\epsi,N,\beta,\xi$ and $\psi^{N,\epsi}(0)$.
Thus Gr\"onwall's lemma proves Theorem~\ref{thm:thm1} once we show that for some $\xi>0$ all terms in the bracket besides $\alpha_\xi(t)$ vanish in the limit $(N,\epsi )\to ( \infty,0)$ along any admissible and moderately confining sequence $( N,\epsi)$. This is true for ${\mu}/{\epsi}$ and  ${ \epsi^2}/{\mu^{\frac32}} = \sqrt{{ \epsi^4}/{\mu^{3}}}$
by assumption.
Since
\[
\frac{\epsi^4}{\mu^3} = \epsi^{4-6\beta} N^{3\beta} \rightarrow 0 \quad\mbox{implies}\quad \epsi   N^{ \frac{3\beta}{4-6\beta}}\rightarrow 0\,,
\]
we have that 
\[
N^\xi \epsi =  \left(N^\xi N^{-\frac{3\beta}{4-6\beta}} \right) \left(\epsi   N^{ \frac{3\beta}{4-6\beta}}\right)\to 0 \quad \mbox{for} \quad 
0<\xi \leq \frac{3\beta}{4-6\beta}
\]
and
\[
N^\xi \frac{a}{\mu^3} =N^\xi \epsi^{2-6\beta} N^{3\beta-1}=  \left( N^\xi N^{-\frac{3\beta(2-6\beta)}{4-6\beta}}\right) \left(\epsi   N^{ \frac{3\beta}{4-6\beta}}\right)^{2-6\beta} N^{3\beta-1} \to 0 \quad \mbox{for} \quad 
0<\xi \leq \frac{2-6\beta}{2-3\beta}\,.
\]
Thus in the case of moderate confinement 
\[
\lim_{( N,\epsi)\to ( \infty,0)} \alpha_\xi(t) = 0
\]
 follows by Gr\"onwall's lemma for $0<\xi\leq \min\left\{ \frac{3\beta}{4-6\beta}, \frac{2-6\beta}{2-3\beta}
\right\}$ and thus with Lemma~\ref{lem:equi} also Theorem~\ref{thm:thm1}. Analogously the statement for strong confinement follows for $0<\xi\leq \frac{3\beta}{4-6\beta}
$.
\end{proof}

\section{Proofs of the Propositions}

\subsection{Preliminaries}\label{sec:beta}

In this section we prove several lemmata that will be used repeatedly in the proofs of the propositions. The first ones are concerned with properties of the operators $\widehat f$  that are at the basis of the condensation-measures $\alpha_f$ (see Definition~\ref{hutdef}).   One should keep in mind, that they are defined with respect to some orthogonal 
projection $p $ in the one-particle space $L^2(\Omega)$. While the first lemma is purely algebraic and holds for general $p$, later on $p = |\ph\rangle\langle\ph|$ will always be the projection onto the one-dimensional subspace spanned by the condensate vector $\ph\in L^2(\Omega)$.
\begin{defn}
 For  $j \in \Z $ we define the shift operator on a function $f:\{0,\cdots, N  \} \to \R$ by
      \begin{align*}
	(\tau_j f)(k) = f(k+j),
      \end{align*}
where we set $(\tau_j f)(k)=0 $ for $k+j \notin \{0, \dots , N \} $.
\end{defn}
\begin{lem}\label{lem:weights} Let $f,g: \{0,\cdots, N  \} \rightarrow \R $, $j\in \{1,\dots, N\}$, and $k\in \{0,\dots, N\}$.
\begin{enumerate}[(a)] 
 \item \label{a} It holds that 
\begin{align*}
 \widehat{f}\,\widehat{g}=\widehat{fg}=\widehat{g}\,\widehat{f}\,, \qquad \widehat f \,p_j = p_j \widehat f\,, \qquad \widehat f \,q_j = q_j \widehat f \,, \quad \mbox{and }\quad  \widehat f P_{k }= P_{k } \widehat f\,.
\end{align*}

\item \label{c} Let $\phi,\psi \in L^2_+(\Omega^N) $  be symmetric and $n(k) = \sqrt{k/N}$, then  
\begin{align*}
\left\llangle \phi, \widehat f \,q_j \psi \right\rrangle &=\left\llangle \phi, \widehat f\, \widehat n^2 \psi \right\rrangle \,.
\end{align*}
If, in addition, $f$ is non-negative, then for $i\in \{1,\dots N\}$, $i \neq j$, it holds that 
\begin{align*}
\left\llangle \psi, \widehat f q_i q_j \psi \right\rrangle &\leq \tfrac{N}{N-1} \left\llangle \psi, \widehat f \,\widehat n^4 \psi \right\rrangle .
\end{align*}

 \item \label{lem:weightsc} Let $T:L^2(\Omega)^{\otimes N}\to L^2(\Omega)^{\otimes N}$ be a bounded operator that acts only on the factors $i$ and $j$ in the tensor product, e.g.\ the two-body potential $w_{ij}$. Then for    $Q_0 := p_i p_j$, $Q_1 \in \{ p_i q_j, q_i p_j\}$, and $ Q_2:= q_i q_j $ we have  
\begin{align*}
\widehat f \,Q_\nu T Q_\mu = Q_\nu T  Q_\mu \,\widehat{\tau_{\nu-\mu} f} 
 \\
 Q_\nu T  Q_\mu\, \widehat f=\widehat { \tau_{\mu-\nu} f} \,Q_\nu T Q_\mu  
 \,.
\end{align*}
\end{enumerate}
\end{lem}

\begin{proof}
 \begin{enumerate}[(a)]
  \item  All commutation relations follow immediately from the definitions. E.g.\
    \begin{align*}
    \widehat f \, \widehat g= \sum_{k,l} f(k)g(l) P_{k }  P_{l } = \sum_{k } f(k)g(k)P_k = \widehat{fg}= \widehat g\, \widehat f.
    \end{align*}

\item
For the equality we find using the symmetry of $\psi$ and $\phi$ and \eqref{Prel} that
\[
\llangle \phi, \widehat f \,q_j \psi \rrangle  =\frac{1}{N}\sum_{j=1}^N \llangle \phi, \widehat f \,q_j \psi \rrangle  =\sum_{k=0}^N\sum_{j=1}^N \frac{f(k)}{N} \llangle \phi,    q_j P_k\psi \rrangle  =\sum_{k=0}^N f(k) \frac{k}{N}\llangle \phi,     P_k\psi \rrangle  = \llangle \phi, \widehat f\, \widehat n^2 \psi \rrangle
\]
For the proof of the inequality let without loss of generality $i=1,j=2$. Then
\begin{eqnarray*}
\llangle \psi, \widehat f q_1 q_2 \psi \rrangle  &=& \tfrac{1}{N(N-1)} \sum_{i \neq j}\llangle \psi, \widehat f  q_i q_j \psi \rrangle 
\stackrel{f\geq0}{ \leq} \tfrac{1}{N(N-1)}  \sum_{i , j}\llangle \psi, \widehat f  q_i q_j \psi \rrangle  
 \\& =&   \tfrac{1}{N(N-1)}  \sum_{k=0}^N\sum_{i,j=1}^N  f(k)  \llangle \psi, q_i q_j   P_k\psi \rrangle\stackrel{\eqref{Prel}}{=}  \tfrac{N^2}{N(N-1)}  \sum_{k=0}^N  f(k)  \frac{k^2}{N^2} \llangle \psi,     P_k\psi \rrangle\\&=&\tfrac{N}{(N-1)} \llangle \psi, \widehat f \widehat n^4 \psi \rrangle .
\end{eqnarray*}
\item Let without loss of generality $i=1$ and $j=2$, 
and let $P^{12}_k$ be the operator $P_k:= \id\otimes \id\otimes P_{k,N-2}$, where $P_{k,N-2}$ is the operator $P_k$ defined on $L^2(\Omega^{N-2})$. Then
\begin{eqnarray*}
\widehat f \,Q_\nu T Q_\mu &=& \sum_{k=0}^N f(k) P_k \,Q_\nu TQ_\mu = \sum_{k=\nu}^{N-2+\nu}f(k) P_k \,Q_\nu TQ_\mu\\&=&
 \sum_{k=\nu}^{N-2+\nu}f(k) P_{k-\nu}^{12} \,Q_\nu TQ_\mu =  \sum_{k=\nu}^{N-2+\nu}f(k) \,Q_\nu TQ_\mu\,P_{k-\nu}^{12} 
 \\&=&
  \sum_{k=\nu}^{N-2+\nu}f(k) \,Q_\nu TQ_\mu\,P_{k-\nu+\mu}  = \sum_{k'=\mu}^{N-2+\mu}f(k'+(\nu-\mu)) \,Q_\nu TQ_\mu\,P_{k'}  \\&=& Q_\nu T Q_\mu\, \widehat{\tau_{\nu-\mu} f} \,
\end{eqnarray*}  
and the converse direction follows in the same way.
\end{enumerate}
\end{proof}

From now on $P_k$ and the derived operations\,\, $\widehat{}$\,\, and $\alpha$ refer to the   projection $p = |\ph\rangle\langle\ph|$ onto the one-dimensional subspace spanned by the one-particle wave function $\ph\in L^2(\Omega)$. We make this explicit only within the following lemma.

\begin{lem}\label{hat.}
Let $\varphi(t)= \Phi(t) \chi $, where $\chi $ is an eigenfunction of $-  \Delta_y + V^\perp$ on $\Omega_\mathrm{f}$ and 
$\Phi(t)$ a solution to  \eqref{equ:grosspqwg} with $\Phi_0\in H^2(\R)$. Then for all $f:\{0,\ldots,N\}\to \R$

\begin{enumerate}[(a)]
 \item   
$
  P_k^{\ph(\cdot)}  \in C^1(\R, \mathcal{L}(L^2(\Omega^N)  ) )$ for all  $k\in \{0, \dots ,N\}$ and thus also $\widehat f^{\ph(t)} \in C^1(\R, \mathcal{L}(L^2(\Omega^N)  ) )$,
 
\item 
$
 \left[   -\Delta_{y_i} + V^\perp(y_i),\widehat f^{\ph(t)}\,\right]=0$ for all $ i\in \{1, \dots ,N\}\,,
$

\item Let $H^\Phi(t):= \sum_{i=1}^N h_i^\Phi(t) $ where $h_i^\Phi(t) $ denotes the one-particle operator   $h^\Phi(t)$ (c.f.\   \eqref{equ:grosspqwg}) acting on the $i$th factor in $L^2(\Omega^N)$. Then 
\[
  \im\frac{\D}{\D t} \widehat f^{\ph(t)} =\left[H^\Phi(t), \widehat f^{\ph(t)}\,\right]\,.
\]
\end{enumerate}
\end{lem}

\begin{proof}
\begin{enumerate}[(a)]
 \item 
This follows immediately from     $\varphi \in C^1(\R, L^2(\Omega))$.
\item This is the fact that a self-adjoint operator commutes  with its spectral projections.   
\item
Because of \eqref{equ:grosspqwg} the projection $|\Phi(t)\rangle\langle\Phi(t)|$ satisfies the differential equation
$\im \partial_t |\Phi(t)\rangle\langle\Phi(t)|= \left[h^\Phi(t) ,|\Phi(t)\rangle\langle\Phi(t)| \right]$ and thus
 $\im \partial_t p_i(t)= \left[h^\Phi_i(t) ,p_i(t) \right]$ and $\im \partial_t q_i(t)= \left[h^\Phi_i(t),q_i(t) \right]$.
The product rule then implies  for any $J\subset \{1,\ldots,N\}$ that 
\[
 \im \, \partial_t \prod_{j\in J} q_j(t) \prod_{j\notin J} p_j(t)  = \sum_{j=1}^N \Big[ h_j^\Phi(t), \prod_{j\in J} q_j(t) \prod_{j\notin J} p_j(t) \Big] =  \Big[ H^\Phi(t), \prod_{j\in J} q_j(t) \prod_{j\notin J} p_j(t)\Big] \,.
\]
As $\widehat f^\ph$ is a linear combination of operators of the above form, the claim follows.
\end{enumerate}
\end{proof}

For the next lemma recall the definition \eqref{mdef} of the function $m(k)$ defining our weight $\alpha_{m}$.
 Because of Lemma~\ref{lem:weights} (c) and the form of the terms I--IV in Proposition~\ref{lem:beta.g}, the difference $m_\ell(k)$ defined below will appear many times in our estimates.
\begin{lem}\label{lem:qs&N}
Let $\psi\in L^2_+(\Omega^N)$ be symmetric, $\ell \in \N$ and 
\begin{equation}\label{melldef}
m_\ell(k) := N (m(k) - \tau_{-\ell}m(k) )= N (  m(k)-  m(k-\ell))\,,
\end{equation}
where the function $m(k)$ was defined in \eqref{mdef}.
\begin{enumerate}[(a)]
\item It holds that
\begin{align*}
0\leq  m_\ell(k)  \leq \begin{cases}
                      \ell\sqrt{\frac{ N}{k}}  \qquad  &k\geq N^{1-2\xi}+\ell\\
		    \frac{\ell}{2}  N^{\xi} \qquad &k < N^{1-2\xi}+\ell  
                    \end{cases}\,,
 \end{align*}
and
\[
 \norm { \widehat{m_\ell} q_1 \psi} \leq \ell\,\|\psi\|  \quad\mbox{and}\quad   \norm { N\big(\widehat{n}-\widehat{\tau_{-\ell} n}\big) q_1 \psi} \leq \ell\,\|\psi\|\,.
\]
\item
Let $q^\chi := {\bf 1}_{L^2(\R)}\otimes(  {\bf1}_{L^2(\Omega_{\rm f})}- |\chi\rangle\langle\chi|)$ be the projection onto the orthogonal complement of the ground state in the confined direction. Then
\[
\left\llangle \psi, q^\chi_1 \psi \right\rrangle \lesssim  \epsi^2\left(1+ |E^\psi (t)  |\right) 
\]  
and
\[
 \left\| \widehat{m_1}\, q^\chi_1 \psi   \right\|\lesssim N^\xi \,\epsi  \, \left(1+ |E^\psi (t)  |\right)^\frac12 \,.
\]
\end{enumerate}
\end{lem}

\begin{proof}
First recall that   $n$ and $m$ are monotonically increasing functions, c.f.~\eqref{mdef}. Moreover
\begin{align*}
 (n(k)-n(k-\ell))^2= \Big  (\tfrac{\sqrt{k}-\sqrt{ k-\ell }}{\sqrt{N}}\Big )^2= \tfrac{\ell^2}{(\sqrt{k}+\sqrt{ k-\ell })^2N} \leq \tfrac{\ell^2}{ k N}
\end{align*}
and thus also $m_\ell(k) \leq \ell\sqrt{\frac{N}{k}}$ for $k\geq N^{1-2\xi}+\ell$ follows. The bound 
$m_\ell(k)\leq  \frac{\ell}{2}  N^{\xi}$ is obvious for $k < N^{1-2\xi}$ and holds also for $  N^{1-2\xi}\leq k<N^{1-2\xi}+\ell$, 
since $\sqrt{\frac{k}{N}} \leq \frac{1}{2}(N^{-1+\xi}k+N^{-\xi})$ for such $k$.

For any $f:\{0,\ldots,N\}\to \R$ we find with Lemma~\ref{lem:weights} (b) that
\begin{align*}
\norm { \left(\widehat{f}-\widehat{\tau_{-l} f}\right) q_1 \psi}^2= \left\llangle \psi,  \left(\widehat{f}-\widehat{\tau_{-l} f}\right)^2 q_1 \psi\right\rrangle=
\sum_{k=1}^N \Big(f(k)-f(k-l)\Big)^2\frac{k}{N}\left\llangle \psi, P_{k } \psi \right\rrangle\, .
\end{align*}
Hence part (a) of the lemma follows with the above estimates on $m_\ell$ and the identity \eqref{Prel}. 
From
\begin{eqnarray*}\lefteqn{
E^\psi (t)   =  \tfrac{1}{N}\left\llangle \psi, \left( H(t)- N\tfrac{E_0}{\epsi^2}\right) \psi\right\rrangle} \\
& =& 
\left\llangle \psi, \frac{1}{N}\left( \sum_{i=1}^N  \left(
\left(\partial_{x_i}  +\theta'(x_i)   L_i \right)  \rho_\epsi^{-2}(r_i)  \left(\partial_{x_i}  +\theta'(x_i)   L_i \right)
- V_{\rm bend}(r_i)+V(r_i)\right.\right.\right.\\
&&\hspace{2cm}  \;+\left.\tfrac{1}{\epsi^2}(-\Delta_{y_i}+V^\perp(y_i)-E_0)\right)\left. + \frac{a}{\mu^3}\sum_{j < i }   w  \left( \frac{ f_\theta^\epsi(r_i)-f_\theta^\epsi(r_j)}{\mu}\right) \Bigg)\psi \right\rrangle \\
&\weq{{\rm \bf A2}}{ \geq} &    \left\llangle \psi,\left(- V_{\rm bend}(r_1) + \tfrac{1}{\epsi^2}\left(-\Delta_{y_1}+V^\perp(y_1 )-E_0\right)\right) \psi \right\rrangle\\
& =& -\left\llangle \psi,    V_{\rm bend}(r_1)\psi \right\rrangle  +  \left\llangle \psi, \tfrac{1}{\epsi^2}\left(-\Delta_{y_1}+V^\perp(y_1)-E_0\right)q_1^\chi  \psi \right\rrangle\\
& \geq &-\left\|    V_{\rm bend} \right\|_{L^\infty(\Omega)}+ \tfrac{C}{\epsi^2}   \left\llangle \psi, q_1^\chi \psi \right\rrangle
\end{eqnarray*}
we infer that    $\llangle \psi, q^\chi_1 \psi \rrangle \lesssim \epsi^2\left(1+ |E^\psi (t) |\right)$. 
For the proof of the remaining estimate   we use that $q_1^\chi$ commutes with $q_1$ and  thus also   with~$P_k$. 
Hence
\begin{eqnarray*} 
 \norm{\widehat{m_1}\, q_1^\chi \psi  }^2 &=& N^2 \sum_{k=1}^N   m_1^2(k)   \llangle  \psi,  P_k \, q_1^\chi\psi \rrangle \\
&\leq&
 \tfrac{1}{4}\sum_{k=1}^{ \lfloor N^{1-2\xi} \rfloor}        N^{ 2\xi}\,  \left\llangle \psi, P_{k} q_1^\chi\psi \right\rrangle
 + \tfrac{1}{4} \sum_{k= \lceil N^{1-2\xi} \rceil}^N   \frac{  N  }{k }  \left\llangle \psi,  P_k q_1^\chi\psi \right\rrangle\\
&\leq& \tfrac12 N^{2\xi}\sum_{k=1}^{ N}  \llangle \psi,    P_k\,q^\chi_1 \psi \rrangle =  \tfrac12 N^{2\xi}  \llangle \psi, q_1^\chi\psi\rrangle \lesssim N^{2\xi} \,\epsi^2\, \left(1+ |E^\psi (t)  |\right)\,.
\end{eqnarray*}
\end{proof}

The meaning of Lemma~\ref{lem:qs&N} (b) is the following.
 Due to symmetry of the wave function, the ``probability'' that one specific particle in a many-body state gets excited in the confined direction can be controlled by the renormalised    energy per particle $E^\psi(t)$ for any $t\in\R$. 

Next we   Taylor expand   the scaled two-body interaction $w^{\epsi, \beta,N}_{12}$.
\begin{lem}\label{R2Lemma} Assuming {\bf A2} for the two-body potential $w$ and {\bf A1} for the geometry of the waveguide, it holds that 
\begin{eqnarray*}
 \frac{\epsi^2}{\mu^3} \; w  \left( \frac{ f_\epsi(r_1)-f_\epsi(r_2)}{\mu}\right) &=& 
  \frac{\epsi^2}{\mu^3}\; w  \left( \frac{ r^\epsi_1-r^\epsi_2}{\mu}\right) \,+\,  R(r_1 , r_2 )   \, \frac{\epsi^2}{\mu^3}\; 
 \tilde w'\left( \tfrac{\left\| r^\epsi_2-r^\epsi_1  \right\|^2}{\mu^2}    \right) \, +\,  \frac{\epsi^2}{\mu^3} \,\Or\left( R(r_1 , r_2 )^2 \right)  \\[1mm]
&=: & w^0_{12}(r_1,r_2) + T_1(r_1 , r_2 ) + T_2(r_1 , r_2 )
\end{eqnarray*}
with
\[
\overline R := \sup_{r_1,r_2\in\Omega} | R(r_1 , r_2 )| \lesssim  \epsi +\mu \,.
\]
\end{lem}

\begin{proof}
The proof is in essence a Taylor expansion, but we need to be careful with the different scalings. First 
recall the maps $f$, $T_\theta$ and $D_\epsi$ defined in the introduction.
$D_\epsi$ is linear and for the differentials of $f$ and $T_\theta$ one easily computes
\begin{eqnarray*}
DT_\theta(x,y) &=& \begin{pmatrix} 1 & 0 && 0 \\ T'_{\theta(x)}y&& T_{\theta(x)}
\end{pmatrix}\\[2mm]
Df  (r) h &=&   (c'(x),  e_1(x),  e_2(x)) h \;+\;  (y^1 e'_1(x)+y^2 e'_2(x)) \,h_x\\
&=:& A(x) h - y\cdot\kappa(x)c'(x) h_x\,. 
\end{eqnarray*}
For $
f_\theta := f\circ T_\theta$ we thus obtain
\begin{eqnarray*}
D f_\theta (r) &=& D(f\circ T_\theta)  (r) h =  Df\circ T_\theta(r)  \;DT_\theta(r)  h \\
&=&\left( (A\circ T_\theta)(x) + (b\circ T_\theta)(r)    (1,0,0)^T\right)
 \begin{pmatrix} 1 & 0 && 0 \\ T'_{\theta(x)}y&& T_{\theta(x)}\end{pmatrix}h\\
 &=& A(x) T_\theta(x) h + A(x) \begin{pmatrix} 0 & 0 &  0 \\ T'_{\theta(x)}y&0&0\end{pmatrix}h
 - (T_{\theta(x)}y\cdot \kappa(x)) c'(x) h_x\\
 &=& A(x) T_\theta(x) h + \left( (e_1(x),e_2(x)) T'_{\theta(x)}y - (T_{\theta(x)}y\cdot \kappa(x)) c'(x) 
 \right) h_x\\
 &=:&  A_\theta(x) h + b_\theta(r) h_x
\end{eqnarray*}
Note that $A_\theta(x)$ is an orthogonal matrix for all $x\in\R$ and $\|b(x,y)\|_{\R^3} \lesssim \|y\|_{\R^2}$. Hence for $\|y\|_{\R^2}$ small enough, $D f_\theta (r)$ is invertible and 
\begin{equation}\label{festi}
\big| \left\| f_\theta(r_2) - f_\theta(r_1) \right\|_{\R^3}  -  \left\|  r_2-r_1  \right\|_{\R^3}\big| \lesssim \|y\|_{\R^2} \left\| f_\theta(r_2) - f_\theta(r_1) \right\|_{\R^3} \,.
\end{equation}
Since $f\in C^\infty(\R^3)$, Taylor expansion gives
\begin{eqnarray*}
f_\theta(r_2) - f_\theta(r_1) 
&=& A_\theta\left(\tfrac{x_1+x_2}{2}\right)\,(r_2-r_1)
+\;   b_\theta\left(\tfrac{r_1+r_2}{2}\right) \,(x_2-x_1) + \Or(|r_2-r_1|^3)
\end{eqnarray*}
and thus 
\begin{eqnarray*}
\|f_\theta(r_2) - f_\theta(r_1) \|^2
&=& \left\langle A_\theta\left(\tfrac{x_1+x_2}{2}\right)\,(r_2-r_1 ), A_\theta\left(\tfrac{x_1+x_2}{2}\right)\,(r_2-r_1 )\right\rangle\\
&&
+\; 2  \left\langle A_\theta \left(\tfrac{x_1+x_2}{2}\right)\,(r_2-r_1 ), b_\theta\left(\tfrac{r_1+r_2}{2}\right)  \right\rangle\,(x_2-x_1) \\
&& +\;   \left\langle b_\theta\left(\tfrac{r_1+r_2}{2}\right) ,b_\theta\left(\tfrac{r_1+r_2}{2}\right) \right\rangle\,(x_2-x_1) ^2+ \Or(|r_2-r_1|^3)\\
&=& \left\| r_2-r_1 \right\|^2 \;+\; \widetilde R(r_1,r_2)
\end{eqnarray*}
with 
\[
|\widetilde R(r_1,r_2)|= \Or(\|r_2-r_1\|\,|x_2-x_1|\|y\|_{\R^2}+  |x_2-x_1|^2\|y\|_{\R^2}^2 +\|r_2-r_1\|^3 )\,.
\]
Now recall that 
\[
f_\epsi := f\circ T_\theta\circ D_\epsi\quad\mbox{i.e.}\quad f_\epsi (r) = f_\theta(r^\epsi)\,.
\]
Since  $w\left( \tfrac{1}{\mu} \left( f_\epsi(r_2) - f_\epsi(r_1) \right)\right) \not=0$ only for $\|f_\epsi(r_2) - f_\epsi(r_1)\|<\mu$ and thus according to \eqref{festi} also $\|r^\epsi_2-r^\epsi_1\|<\mu(1+\epsi)$, we have that
\[
|\widetilde R(r_1^\epsi, r_2^\epsi)| =  \Or(\mu^2 \epsi +\mu^2 \epsi^2+ \mu^3 )\,.
\]
Taylor expansion of $w(r) = \tilde w(|r|^2)$ finally gives the desired result with $R:=\widetilde R/\mu^2$, 
\begin{eqnarray*}\lefteqn{\hspace{-0.5cm}
w\left( \tfrac{1}{\mu} \left( f_\epsi(r_2) - f_\epsi(r_1) \right)\right) \;=\; \tilde w\left( \tfrac{1}{\mu^2} \left\| f_\epsi(r_2) - f_\epsi(r_1) \right\|^2\right)}\\
&=& \tilde w\left( \tfrac{1}{\mu^2}\left\| r^\epsi_2-r^\epsi_1   \right\|^2\right)  + \frac{\widetilde R(r_1^\epsi, r_2^\epsi)}{\mu^2} \, \tilde w'\left( \tfrac{1}{\mu^2} \left\| r^\epsi_2-r^\epsi_1  \right\|^2   \right) +   \Or\left(\frac{\widetilde R(r_1^\epsi, r_2^\epsi)^2}{\mu^4}\right)\,.
\end{eqnarray*}

\end{proof}

The next lemma collects elementary facts that will allow us to estimate one- and two-body potentials in different situations.

\begin{lem}\label{lem:young}
Let $g:\R^3\times \R^3\to\R$ be a measurable function such that $|g(r_1,r_2)|\leq v(r_2-r_1)$ almost everywhere for some measurable function $v:\R^3\to \R$, and
 let   $\ph\in L^2(\Omega)\cap L^\infty(\Omega)$ with $\|\ph\|_2=1$, and $p= |\varphi \rangle \langle \varphi |$.   
 \begin{enumerate}[(a)]
  \item For $v\in L^2(\R^3)$ we have \\[1mm]
      $
	\norm{v(r)p}_{\mathcal{L}(L^2(\Omega))} \leq \norm{v}_{L^2(\Omega)}\norm{\varphi}_{L^\infty(\Omega)}
     $
\quad and\quad 
	$
        \norm{g(r_1,r_2)p_1}_{\mathcal{L}(L^2(\Omega^2))} \leq \norm{v}_{L^2(\R^3)} \norm{\varphi}_{L^\infty(\Omega)}
       $
\item  For $v\in L^1(\R^3)$ we have
$
\norm{ p_1 g(r_1,r_2)p_1}_{\mathcal{L}(L^2(\Omega^2))} \leq \norm{v}_{L^1(\R^3)}\norm{\varphi}^2_{L^\infty(\Omega)}
$.
\item For $\Phi\in H^2(\R)$ 
\[
\|\Phi\|_{L^\infty}^2 \leq \|\Phi\|_{H^1}^2 \leq \|\Phi\|_{H^2}^2\qquad\mbox{and}\qquad \|\nabla |\Phi|^2 \|_{L^2}  \leq 2 \|\Phi\|_{L^\infty} \|\Phi\|_{H^1} \,.
\]
\end{enumerate}
\end{lem}

\begin{proof}
All three estimates in (a) and (b) are elementary: 
\begin{eqnarray*}
 \norm{v(r )p }^2_{\mathcal{L}(L^2(\Omega))} &=&\sup_{\norm{\psi}=1} \left\langle \psi, p  |v(r )|^2 p  \psi  \right\rangle_{L^2(\Omega)}  
= \left\langle \varphi(r ) ,  |v(r )|^2  \varphi(r ) \right\rangle_{L^2(\Omega)}   \sup_{\norm{\psi}=1} \| p   \psi  \|^2_{L^2(\Omega)}\\
&\leq & \left\langle v(r ) ,  |\ph(r )|^2  v(r ) \right\rangle_{L^2(\Omega)}  \leq \|v\|_{L^2(\Omega)}^2\, \|\ph\|_{L^\infty(\Omega)}^2\,,
\end{eqnarray*}
  \begin{eqnarray*}
\norm{g(r^\epsi_1,r^\epsi_2)p_1}^2_{\mathcal{L}(L^2(\Omega^2))}&=&\sup_{\norm{\psi}=1} \left \langle p_1  \psi,  |g(r^\epsi_1,r^\epsi_2)|^2 p_1 \psi  \right \rangle_{L^2(\Omega^2)} \\
&\leq& \sup_{\norm{\psi}=1} \|p_1\psi\|_{L^2(\Omega^2)}\; \sup_{r_2\in\Omega} \int_{\Omega } |\ph(r_1)|^2  |v(r^\epsi_2-r^\epsi_1)|^2 \,\D r_1 \; \\
&\leq & \norm{\varphi}_{L^\infty(\Omega)}^2 \;\sup_{r_2\in\Omega} \int_{\Omega }  |v(r^\epsi_2-r^\epsi_1)|^2\, \D r_1 
\\
&\leq &  \norm{\varphi}_{L^\infty(\Omega)}^2  \tfrac{1}{\epsi^2} \norm{v}^2_{L^2(\R^3)}\,,
\end{eqnarray*}
\begin{eqnarray*}
 \norm{p_1 g(r^\epsi_1,r^\epsi_2)p_1}_{\mathcal{L}(L^2(\Omega^2))}  &\leq &     \sup_{r_2\in\Omega} \int_{\Omega } |\ph(r_1)|^2  |g(r^\epsi_1,r^\epsi_2)|  \, \D r_1 \;\leq \;\sup_{r_2\in\Omega} \int_{\Omega } |\ph(r_1)|^2  |v(r^\epsi_2-r^\epsi_1)|  \, \D r_1 \;\\
&\leq &  \norm{\varphi}^2_{L^\infty(\Omega)} \tfrac{1}{\epsi^2} \norm{v}_{L^1(\R^3)}.
\end{eqnarray*}
For (c) note that for $\Phi\in H^2(\R) \subset C^1(\R)$ we have with Cauchy-Schwarz
\[
\overline{\Phi(x)}\Phi(x) = \int_{-\infty}^x \overline{\Phi'(s)}\Phi(s) +\overline{\Phi(s)}\Phi'(s)\,\D s \leq 2 \|\Phi'\|_{L^2(\R)} \|\Phi\|_{L^2(\R)}\leq   \|\Phi'\|_{L^2 }^2 +  \|\Phi\|_{L^2 }^2 =  \|\Phi\|_{H^1 }\,,
\]
and
\[
 \|\nabla |\Phi|^2 \|_{L^2}^2 \leq \int 4 |\Phi'(x)|^2 |\Phi(x)|^2 \,\D x\leq 4 \|\Phi\|_{L^\infty}^2\|\Phi'\|_{L^2 }^2 \leq 4 \|\Phi\|_{L^\infty}^2 \|\Phi\|_{H^1} \,.
\]
\end{proof}

In the following corollary we collect bounds on the two-body interaction that will be used repeatedly. 
\begin{cor}\label{wcor} For the scaled two-body interaction $w^{\epsi, \beta,N}_{12}$ we have that
\begin{eqnarray*}
\left\| w^{\epsi, \beta,N}_{12} p_1\right\|_{\mathcal{L}(L^2(\Omega^2))} &\lesssim&
   \|\Phi\|_{H^2(\R)}\cdot \left\{\begin{array}{ll}  \tfrac{\epsi}{\mu^{3/2}} & \mbox{for moderate confinement}\\
\sqrt{\mu}& \mbox{for strong confinement}
\end{array}\right.
\\
\left\| \sqrt{w^{\epsi, \beta,N}_{12}} p_1\right\|_{\mathcal{L}(L^2(\Omega^2))} &\lesssim&  \|\Phi\|_{H^2(\R)}\cdot \left\{\begin{array}{ll}  1 & \mbox{for moderate confinement}\\
\sqrt{\mu} & \mbox{for strong confinement}
\end{array}\right.
\\
\|T_1 p_1\|_{\mathcal{L}(L^2(\Omega^2))}+ \|T_2p_1\|_{\mathcal{L}(L^2(\Omega^2))}&\lesssim& \frac{(\epsi+\mu)\epsi}{\mu^{3/2}}\; \|\Phi\|_{H^2(\R)}
 \,,
\end{eqnarray*}
where $T_1$ and $T_2$ are defined in Lemma~\ref{R2Lemma}.
\end{cor}

\begin{proof}
According to Lemma~\ref{lem:young}  (b) and (c) we  have 
\[
\norm{  w^0_{12}  p_1}^2_{\mathcal{L}(L^2(\Omega^2))}\leq \|\ph\|^2_{L^\infty(\Omega)} \,\|w^0_{12} \|_{L^2(\Omega)}^2 \;\lesssim\; \|\Phi\|^2_{L^\infty(\R)}\cdot \left\{\begin{array}{ll}  \tfrac{\epsi^2}{\mu^3} & \mbox{for moderate confinement}\\
\frac{\epsi^4}{\mu^3} & \mbox{for strong confinement}
\end{array}\right.
\]
and 
\[
\norm{p_1 w^0_{12}  p_1}_{\mathcal{L}(L^2(\Omega^2))}\leq \|\ph\|^2_{L^\infty(\Omega)}  \,\|w^0_{12}\|_{L^1(\Omega)}  \;\lesssim\; \|\Phi\|^2_{L^\infty(\R)}  \cdot \left\{\begin{array}{ll}  1 & \mbox{for moderate confinement}\\
\epsi^2 & \mbox{for strong confinement}
\end{array}\right.
\,.
\]
Using the bound
\[
\left|   T_1(r_1^\epsi, r_2^\epsi)\right| \leq 
 \frac{\epsi^2}{\mu^3}  \, \overline R   \,
 \tilde w'\left( \tfrac{\left\| r^\epsi_2-r^\epsi_1  \right\|^2}{\mu^2}     \right)
=: v(r_2^\epsi- r_1^\epsi) 
\]
we obtain analogously
\[
\norm{  T_1 p_1}^2_{\mathcal{L}(L^2(\Omega^2))}\leq \|\ph\|^2_{L^\infty(\Omega)}   \,\|v\|_{L^2(\Omega)}^2 \;\lesssim\; \|\Phi\|^2_{L^\infty(\R)} \tfrac{\epsi^2}{\mu^3}(\epsi+\mu)^2\
\]
and 
\[
\norm{p_1 T_1  p_1}_{\mathcal{L}(L^2(\Omega^2))}\leq \|\ph\|^2_{L^\infty(\Omega)}  \,\|v\|_{L^1(\Omega)}  \;\lesssim\; \|\Phi\|^2_{L^\infty(\R)} (\epsi+\mu) 
\]
for moderate confinement and an additional factor $\epsi^2$ for strong confinement.
With  $\|T_2\|\leq \frac{\epsi^2(\epsi+\mu)^2}{\mu^3}$ and $\|\Phi\|_{L^\infty(\R)}\leq \|\Phi\|_{H^2(\R)}\geq 1$ we arrive at
\[
\left\| w^{\epsi, \beta,N}_{12} p_1\right\|_{\mathcal{L}(L^2(\Omega^2))} \lesssim
  \frac{\epsi}{\mu^{\frac32}}\|\Phi\|_{H^2(\R)} \quad\mbox{and}\quad
\left\| p_1 w^{\epsi, \beta,N}_{12} p_1\right\|_{\mathcal{L}(L^2(\Omega^2))}  \lesssim    \|\Phi\|_{H^2(\R)}^2
\]
for moderate confinement, and 
\[
\left\| w^{\epsi, \beta,N}_{12} p_1\right\|_{\mathcal{L}(L^2(\Omega^2))} \lesssim
 \sqrt{\mu}\,\|\Phi\|_{H^2(\R)} \quad\mbox{and}\quad
\left\| p_1 w^{\epsi, \beta,N}_{12} p_1\right\|_{\mathcal{L}(L^2(\Omega^2))}  \lesssim  \mu\,  \|\Phi\|_{H^2(\R)}^2
\]
for strong confinement.
Finally
\[
\left\| \sqrt{w^{\epsi, \beta,N}_{12}} p_1\right\|_{\mathcal{L}(L^2(\Omega^2))}^2= 
\sup_{\|\psi\|=1} \left\langle \psi, p_1 w^{\epsi, \beta,N}_{12}  p_1 \psi\right\rangle \leq \left\| p_1 w^{\epsi, \beta,N}_{12}  p_1\right\|_{\mathcal{L}(L^2(\Omega^2))} \,.
\]
\end{proof}

Finally we need also the following lemma that shows how to bound  the ``kinetic energy'' of $q_1\psi$ in terms of $\alpha_\xi$. 

\begin{lem}\label{lem:energyestimate}
 Let the assumptions of Theorem~\ref{thm:thm1}  hold.  Then in the moderately confining case  
 \begin{align*}
 \left\| \left(\tfrac{\partial}{\partial x_1} +  \theta'(x_1) L_1 \right) q_1\psi^{N,\epsi}(t) \right\|^2 \;\lesssim     \;\|\Phi(t)\|_{H^2(\R)}^3   \left(\alpha_\xi(t)   +\frac{\mu}{\epsi} + \frac{a}{ \mu^3} \right)\,.
 \end{align*}
\end{lem} 

This energy estimate is actually quite subtle and we postpone its proof to Subsection~\ref{energylemmaproof}.
Note that it is the only place in our argument where the positivity of the interaction is crucial. 

We end this subsection with the proof of Lemma~\ref{lem:equi}.
 \begin{proof}[Proof of Lemma~\ref{lem:equi}]
 $(i) \Leftrightarrow (ii)$: Let  $\lim \alpha_{n^a}(\psi^N,\varphi)=0$ for some $a>0$. Then   $  \alpha_{n^b}(\psi^N,\varphi)\leq \alpha_{n^a}(\psi^N,\varphi)$   for all   $b>a$ since $n^b\leq n^a$. If $\frac{a}{2}\leq b<a$, then
\[
\alpha_{n^b}(\psi^N,\varphi) = \left\llangle \psi^N, \widehat n^b\,\psi^N\right\rrangle = \left\llangle \widehat n^{ b-\frac{a}{2}}  \psi^N, \widehat n^\frac{a}{2}\,\psi^N\right\rrangle \leq \|\widehat n^{ b-\frac{a}{2}}  \psi^N\|\,\|\widehat n^\frac{a}{2}\,\psi^N\|\leq \sqrt{\alpha_{n^a}(\psi^N,\varphi)}\,.
\]
 $(i) \Rightarrow (iii)$: For $a=2$   we have $\lim_{N\to \infty} \|q_1\psi^N\|=0$  according to \eqref{n2comp}. 
Let $P_k^M$ be the projector from Definition~\ref{def:pP} acting on $L^2(\Omega^M)$ with $p=|\ph\rangle\langle\ph|$. 
Then in
\[
\gamma^N_M = \sum_{k=0}^M \sum_{k'=0}^M P_k^M \gamma^N_M P_{k'}^M
\]
 all terms but the one with $k=k'=0$ go to zero in norm for $N\to\infty$. Hence,
 \[
 \lim_{N\to \infty} \gamma^N_M =  \lim_{N\to \infty}  p^{\otimes M} \,\gamma^N_M \,p^{\otimes M} = p^{\otimes M}\,,
 \]
 where the last equality follows from Tr$\,\gamma^N_M\equiv 1$ and the fact that $p^{\otimes M}$ has rank one. 
 
 $(iii) \Rightarrow (iv)$:  
 We learned the following argument from \cite{RodSch07}. Since $p^{\otimes M}$ has rank one, the operator $\gamma^N_M - p^{\otimes M}$ can have at most one negative eigenvalue $\lambda_-<0$. Since  
 Tr$\,(\gamma^N_M - p^{\otimes M})=0$,  $|\lambda_-|$ equals the sum of all positive eigenvalues. Hence 
 \[
 {\rm Tr}\left|\gamma^N_M - p^{\otimes M}\right| = 2\left|\lambda_-\right| = 2 \left\|\gamma^N_M - p^{\otimes M}\right\|\,.
 \]
  $(iv) \Rightarrow (v)$ is obvious and  $(v) \Rightarrow (i)$ follows for $a=2$ from 
  \begin{eqnarray*}
\alpha_{n^2}(\psi^N,\ph)& \stackrel{ \eqref{n2comp}}{=}& \left\llangle  \psi, (1-p_1)\psi\right\rrangle = {\rm Tr}\left(p- p\gamma^N_1p \right) = {\rm Tr}\left|p - p\gamma^N_1p \right| = {\rm Tr}\left|p \left(p-  \gamma^N_1\right) \right|\\& \leq& \|p\|  {\rm Tr}\left|p -  \gamma^N_1  \right|\,.
  \end{eqnarray*}
 \end{proof}

\subsection{Proof of Proposition~\ref{lem:beta.g}}
Recalling the definition \eqref{alphadef} we need to estimate 
\[
\left| \tfrac{\D}{\D t}  \alpha_\xi(t)\right|\leq \left|\tfrac{\D}{\D t}  \alpha_{m}(\psi^{N,\epsi}(t) ,\ph(t)) \right|+\left| \tfrac{\D}{\D t}  |E^{\psi^{N,\epsi}(t)}(t)-E^{\Phi(t)}(t)| \right|\,.
\]
For better readability we abbreviate $\psi = \psi^{N,\epsi}(t)$ and $\Phi=\Phi(t)$ in the remainder of this proof. 
The derivative of the second term yields 
\begin{align*}
\left| \tfrac{\D}{\D t}  |E^{\psi }(t)-E^{\Phi}(t)| \right|= |\llangle \psi, \dot V(t,x_1,\epsi y_1) \psi \rrangle - \langle \Phi, \dot V(t,x_1,0) \Phi \rangle_{L^2(\R)}| \,. 
\end{align*}
As of Lemma\,\ref{hat.}, the map $t\mapsto \alpha_{m}(\psi,\ph) \in C^1(\R,\R)$ and  we find
 \begin{eqnarray}
 \tfrac{\D}{\D t} \alpha_{m}&= &\tfrac{\D}{\D t} \left\llangle \psi, \widehat m \psi \right\rrangle  \stackrel{\ref{hat.}}{ =} \im \left\llangle \psi, \left[H^\epsi_N-H^\Phi, \widehat m\right] \psi \right\rrangle\nonumber\\
&\stackrel{\mathclap{\ref{hat.}}}{ =} & \im \left\llangle \psi,   \Big[ \tfrac{1}{N-1}\sum_{i< j} w^{\epsi, \beta,N}_{ij}-\sum_{i=1}^N b|\Phi(x_i)|^2 , \widehat m\Big] \psi \right\rrangle+ \im \left\llangle \psi,  \Big[\sum_{i=1}^N    V(x_i,\epsi y_i)- \sum_{i=1}^N V(x_i,0) , \widehat m\Big] \psi \right\rrangle \nonumber \\
&&
+\;\im \left\llangle \psi,  \Big[\sum_{i=1}^N    - (\theta'(x_i)L_i)^2 -  |\theta'(x_i)|^2 \,\|L\chi\|^2 , \widehat m\Big] \psi \right\rrangle\nonumber   +    \im \left\llangle \psi,  \Big[\sum_{i=1}^N   R^{(1)}_i , \widehat m\Big] \psi \right\rrangle\nonumber 
\\
& = & \tfrac{ \im N}{2}   \left\llangle \psi,   \left[   w^{\epsi, \beta,N}_{12}- b|\Phi(x_1)|^2-   b|\Phi(x_2)|^2 , \widehat m\right] \psi \right\rrangle\;+\;\im N
    \left\llangle \psi,    \left[   V(x_1,\epsi y_1)-V(x_1,0) , \widehat m\right] \psi \right\rrangle  \nonumber\\
    &&-\; \im N
    \left\llangle \psi,    \left[   (\theta'(x_1)L_1)^2 +  |\theta'(x_1)|^2 \,\|L\chi\|^2 , \widehat m\right] \psi \right\rrangle
    +\im N
    \left\llangle \psi,    \left[   R^{(1)}_1 , \widehat m\right] \psi \right\rrangle\nonumber
    \\
&= &\tfrac{\im}{2} N \left\llangle \psi,(p_1+q_1)(p_2+q_2)  \left[    w^{\epsi, \beta,N}_{12}-  b|\Phi(x_1)|^2-   b|\Phi(x_2)|^2 , \widehat m\right](p_1+q_1)(p_2+q_2) \psi \right\rrangle\label{firstsummand}\\
&& + \;\im N \left\llangle \psi, (p_1+q_1)   \left[   V(t,x_1,\epsi y_1)-V(t,x_1,0)  , \widehat m\right] (p_1+q_1) \psi \right\rrangle\, .\label{secondsummand}\\
&&-\; \im N
    \left\llangle \psi,  (p_1+q_1)   \left[   (\theta'(x_1)L_1)^2 +  |\theta'(x_1)|^2 \,\|L\chi\|^2 , \widehat m\right] (p_1+q_1) \psi \right\rrangle\label{thirdsummand}\\&&
    +\;\im N
    \left\llangle \psi,  (p_1+q_1)   \left[   R^{(1)}_1 , \widehat m\right] (p_1+q_1) \psi \right\rrangle\label{lastsummand}
\end{eqnarray}
According to  Lemma~\ref{lem:weights} (\ref{lem:weightsc}) all terms with the same number of $p$'s and $q$'s on each side of the commutator vanish. Therefore we find 
that \eqref{secondsummand}--\eqref{lastsummand} are bounded by   
\begin{eqnarray*}
|\eqref{secondsummand}+\eqref{thirdsummand}+\eqref{lastsummand}|&\leq &
 2 N   \left|\left\llangle \psi, p_1    \left[   V(t,x_1,\epsi y_1)-V(t,x_1,0) , \widehat m\right] q_1 \psi \right\rrangle\right|\\
 &&
 + \; 2N \left| \left\llangle \psi,   p_1    \left[   (\theta'(x_1)L_1)^2 +  |\theta'(x_1)|^2 \,\|L\chi\|^2 , \widehat m\right] q_1  \psi \right\rrangle \right|\\
 &&
 + \; 2N \left| \left\llangle \psi,   p_1    \left[  R^{(1)}_1 , \widehat m\right] q_1  \psi \right\rrangle \right|
 \,.
\end{eqnarray*}
The crucial step (c.f.\ \cite{Pic08}) is to
split   \eqref{firstsummand}   according to 
\begin{eqnarray*}
 \lefteqn{ \tfrac{\im}{2} N \left\llangle \psi,(p_1+q_1)(p_2+q_2)  \left[     w^{\epsi, \beta,N}_{12}-  b|\Phi(x_1)|^2-   b|\Phi(x_2)|^2 , \widehat m\right](p_1+q_1)(p_2+q_2) \psi \right\rrangle } \\
&=&\tfrac{\im}{2} N \left\llangle \psi, p_1 p_2  \left[    w^{\epsi, \beta,N}_{12}-  b|\Phi(x_1)|^2-   b|\Phi(x_2)|^2 , \widehat m\right] p_1p_2 \psi \right\rrangle\\
&&+\;\tfrac{\im}{2} N \left\llangle \psi, p_1 p_2  \left[    w^{\epsi, \beta,N}_{12}-  b|\Phi(x_1)|^2-   b|\Phi(x_2)|^2 , \widehat m\right] (p_1q_2+q_1p_2 +q_1q_2) \psi \right\rrangle\\
&&+\;\tfrac{\im}{2} N \left\llangle \psi, (p_1 q_2+ q_1 p_2)  \left[    w^{\epsi, \beta,N}_{12}-  b|\Phi(x_1)|^2-   b|\Phi(x_2)|^2 , \widehat m\right](p_1p_2+q_1q_2) \psi \right\rrangle\\
&& +\;\tfrac{\im}{2}N \left\llangle \psi, q_1 q_2  \left[    w^{\epsi, \beta,N}_{12}-  b|\Phi(x_1)|^2-   b|\Phi(x_2)|^2 , \widehat m\right](p_1p_2+ q_1p_2 + p_1q_2) \psi \right\rrangle\\
&\weq{\mathrm{sym.}}{=}&  \im N \left\llangle \psi, p_1 p_2  \left[   w^{\epsi, \beta,N}_{12}-  b|\Phi(x_1)|^2-   b|\Phi(x_2)|^2 , \widehat m\right] p_1q_2 \psi \right\rrangle + c.c.\\
&&+\;\tfrac{\im}{2}  N \left\llangle \psi, p_1 p_2  \left[     w^{\epsi, \beta,N}_{12} -  b|\Phi(x_1)|^2-   b|\Phi(x_2)|^2 , \widehat m\right] q_1q_2 \psi \right\rrangle +c.c.\\
&& +\;\im N \left\llangle \psi, p_1 q_2  \left[     w^{\epsi, \beta,N}_{12}-  b|\Phi(x_1)|^2-   b|\Phi(x_2)|^2 , \widehat m\right] q_1q_2 \psi \right\rrangle + c.c.\\
&=&   \im N \left\llangle \psi, p_1 p_2  \left[    w^{\epsi, \beta,N}_{12}-  b|\Phi(x_2)|^2  , \widehat m\right] p_1q_2 \psi \right\rrangle + c.c.\\
&& +\;\tfrac{\im}{2}  N \left\llangle \psi, p_1 p_2  \left[     w^{\epsi, \beta,N}_{12}  , \widehat m\right] q_1q_2 \psi \right\rrangle +c.c.\\
&& +\;\im N \left\llangle \psi, p_1 q_2  \left[     w^{\epsi, \beta,N}_{12}-     b|\Phi(x_1)|^2 , \widehat m\right] q_1q_2 \psi \right\rrangle + c.c.\\
&=&    - 2 \Im \mathrm I-   \Im  \mathrm {II} -  2 \Im  \mathrm {III}.
\end{eqnarray*}
The term with $p_1p_2$ on both sides of the commutator vanishes again because of Lemma~\ref{lem:weights} (\ref{lem:weightsc}).
For the second to last equality we used that for $i\not=j$ the projection $p_j$ commutes with $\widehat m$ and with $|\Phi(x_i)|^2$ and that $p_jq_j=0$.

\subsection{Proof of Proposition~\ref{lem:3termeg}}
 
\begin{proof}[Proof of the bound for {\rm I}]

In the case of moderate confinement, the term I is small due to the cancellation of the mean field  and the full interaction. 
Since $b|\Phi|^2 $ is the mean field for a condensate in the state $\Phi\chi$, i.e.\ a condensate that is in the  ground state with respect to the confined directions, this cancellation works only for the part of $\psi$ that is in this confined ground state. We thus need to split $\psi$ accordingly and introduce the following projections on $L^2(\Omega)$,
\[
p ^\chi :=1 \otimes |\chi  \rangle \langle \chi  | \,,\quad  q^\chi  =1-p^\chi \,, \quad 
p ^\Phi = |\Phi  \rangle \langle \Phi  | \otimes 1 \,,\quad   q^\Phi   =1-p^\Phi \,.
\]
As in Definition~\ref{def:pP} we also introduce the projections $p^\chi_j$, $q^\chi_j$, $p^\Phi_j$, and $q_j^\Phi$ on $L^2(\Omega^N)$.
With these projections we can rewrite  
\begin{equation}\label{equ:qinqx}
q_j= 1- p_j = 1 - p^\Phi_j p^\chi_j = (1-p_j^\chi)+ (1-p_j^\Phi)p_j^\chi=q_j^\chi+q_j^\phi p_j^\chi\,,
\end{equation}
where we recall that $p_j := p_j^\ph$.
Now with  Lemma~\ref{lem:weights}, \eqref{equ:qinqx} and \eqref{melldef} we find
\begin{eqnarray}
 |I|&=&N \left|\left\llangle \psi, p_1 p_2  \left[    w^{\epsi,\beta,N}_{12}-   b|\Phi|^2(x_2) , \widehat {m}\right] p_1q_2 \psi \right\rrangle\right|\nonumber\\
& \weq{\ref{lem:weights}}{=} &N\left|\left\llangle  \psi, p_1 p_2 \left( w^{\epsi,\beta,N}_{12}-  b|\Phi|^2(x_2)\right) \left({\widehat { {m}}}- { \widehat  {\tau_{-1} m}}  \right) p_1 q_2   \psi \right\rrangle  \right|\nonumber\\
&\weq{\eqref{melldef},\eqref{equ:qinqx}}{=} &\;\left|\left\llangle  \psi, p_1 p_2 \left( w^{\epsi,\beta,N}_{12}-  b|\Phi|^2(x_2)\right) \,\widehat{m_1}\,  p_1 \left( p_2^\chi q_2^\Phi+ q_2^\chi \right)    \psi \right\rrangle \right|\nonumber\\
&\leq & \left|\left\llangle  \psi, p_1 p_2 \left( w^{\epsi,\beta,N}_{12}-  b|\Phi|^2(x_2)\right) p_1   p_2^\chi \,\widehat{m_1}\, q_2^\Phi      \psi \right\rrangle\right|
+ \left|\left\llangle  \psi, p_1 p_2  \,w^{\epsi,\beta,N}_{12} \,\widehat{m_1}\, p_1   q_2^\chi      \psi \right\rrangle\right|\nonumber
 \\
&\stackrel{\ref{R2Lemma}}{=}& \left|\left\llangle  \psi, p_1 p_2 \left( w^0_{12}-  b|\Phi|^2(x_2)\right) p_1   p_2^\chi \,\widehat{m_1}\, q_2^\Phi      \psi \right\rrangle\right| \label{I1term}
\\
&&+\; \left|\left\llangle  \psi, p_1 p_2 \left(T_1+T_2 \right) p_1   p_2^\chi \,\widehat{m_1}\, q_2^\Phi      \psi \right\rrangle\right|\label{I3term}\\
&& +\; \left|\left\llangle  \psi, p_1 p_2  \,w^{\epsi,\beta,N}_{12} \,\widehat{m_1}\, p_1   q_2^\chi      \psi \right\rrangle\right|
\label{I2term}\,.
\end{eqnarray}
In the first term \eqref{I1term} the interaction $w^0_{12}$ acts between states that are fixed in the $r_1$ and the $y_2$ variable, so only a $x_2$-dependence remains that approximately cancels the mean field $b|\Phi(x_2)|^2$. More precisely, 
between $p_1p_2$ and $p_1 p_2^\chi$  the leading part $w^0_{12}$ of the interaction   can be replaced by the effective potential
\begin{eqnarray}\lefteqn{
\left\langle \ph\otimes\chi, \frac{\epsi^2}{\mu^3} w \left(\frac{r^\epsi_1-r^\epsi_2}{\mu}\right) \ph\otimes \chi\right\rangle_{L^2(\Omega\times \Omega_{\rm f})} (x_2)=}\nonumber\\
&=&
\frac{\epsi ^2}{\mu^3} \int |\Phi(x_1)|^2\, |\chi(y_1)|^2 \,w\left(\mu^{-1}( (x_1-x_2), \epsi(y_1-y_2))\right) \,|\chi(y_2)|^2\,\D x_1\D y_1 \D y_2\nonumber\\
&=& \int |\chi(y_2)|^2\, \left(
\frac{\epsi ^2}{\mu^3} \int   |\Phi(x_2-x)|^2\, |\chi(y_2-y)|^2 \,w\left(\mu^{-1}( x , \epsi \, y )\right) \,\D x\, \D y \right) \D y_2\label{wexpansion}\,.
\end{eqnarray}
To see that this is close to $b|\Phi(x_2)|^2$, first note 
 that for $f\in C^\infty_0(\Omega)$ we have with $z := (x,y)$ that
\begin{eqnarray*}\lefteqn{
\frac{\epsi ^2}{\mu^3} \int  f(z_2-z) \,w\left(\mu^{-1}( x , \epsi \, y )\right) \,\D x\, \D^2 y}\\ & = & f(z_2)\,\|w\|_{L^1(\R^3)} \;-\; \frac{\epsi ^2}{\mu^3} \int  \int_0^1 \nabla f(z_2 - s z)\cdot z  \,w\left(\mu^{-1}( x , \epsi \, y )\right) \,\D s\,\D x\, \D^2 y\nonumber\\
&=:& f(z_2)\,\|w\|_{L^1(\R^3)} \;+\; R(z_2)\,,
\end{eqnarray*}
where the $L^2$-norm of the remainder is bounded by
\begin{eqnarray}
\|R\|^2_{L^2(\Omega)} &\leq& \|\nabla f \|^2_{L^2(\Omega)}\,\left(\frac{\epsi^2}{\mu^3}
\int |z | \,w\left(\frac{( x , \epsi \, y )}{\mu}\right) \,\D x\, \D^2 y\right)^2\nonumber\\
&=& 
\|\nabla f \|^2_{L^2(\Omega)}\,\left(\frac{\epsi^2}{\mu^2}
\int \frac{|z |}{\mu} \,w\left(\frac{( x , \epsi \, y )}{\mu}\right) \,\D x\, \D^2 y\right)^2\nonumber\\
&=& 
\|\nabla f \|^2_{L^2(\Omega)}\,\left(\mu\epsi^2
\int {|z |} \,w\left(( x , \epsi \, y )\right) \,\D x\, \D^2 y\right)^2\nonumber\\
&\leq& 
\|\nabla f \|^2_{L^2(\Omega)}\,\left(\mu\epsi^2
\int \frac{|z |}{\epsi} \,w\left(( x ,  y )\right) \,\D x\, \frac{\D^2 y}{\epsi^2}\right)^2\;\leq\; \frac{\mu^2}{\epsi^2} \,\|\nabla f \|^2_{L^2(\Omega)}\, \||z|w(z)\|_{L^1(\R^3)}\,.\nonumber
\end{eqnarray}
Hence
\begin{equation}
\left\| \frac{\epsi ^2}{\mu^3} \int  f(\cdot-z) \,w\left(\mu^{-1}( x , \epsi \, y )\right) \,\D x\, \D^2 y - f\|w\|_{L^1(\R^3)} \right\|_{L^2(\Omega)}\;\lesssim\; \frac{\mu }{\epsi } \,\|\nabla f \| _{L^2(\Omega)}\label{faltest}
\end{equation}
and this bound   extends to $f\in H^1(\Omega)$ by density, in particular, to $f= |\Phi|^2|\chi|^2$. Inserting this bound  with \eqref{wexpansion} into  \eqref{I1term} yields, together with Lemma~\ref{lem:young} (a+d) and Lemma~\ref{lem:qs&N} the bound 
\[
\eqref{I1term} \lesssim \frac{\mu}{\epsi}\norm{\nabla |\Phi|^2}_{L^2(\R)} \norm{\Phi}_{L^\infty(\R)}
\lesssim  \frac{\mu}{\epsi}\,\norm{\Phi}_{H^2(\R)}^3 \,.
\]
For the term \eqref{I3term} we have with Corollary~\ref{wcor} and Lemma~\ref{lem:qs&N} that
\begin{eqnarray*}
 \eqref{I3term}
&\leq& \left( \|T_1p_1\| + \|T_2p_1\|\right) \|\widehat{m_1}\, q_2      \psi\|
\lesssim \frac{(\epsi+\mu)\epsi}{\mu^{3/2}} \|\Phi\|_{H^2(\R)}
\,.
\end{eqnarray*}
The term \eqref{I2term} is small due to energy conservation and the energy gap of order $\epsi^{-2}$ between the ground state and the first excited state in the confined direction.  With the help of  Lemma~\ref{lem:qs&N} we get
\begin{eqnarray*}\lefteqn{\hspace{-1cm}
  \left|\left\llangle \psi, p_1 p_2   w^{\epsi,\beta,N}_{12} \,\widehat{m_1}\,p_1 q_2^\chi \psi \right\rrangle  \right|
\;\leq\; \left\| p_1 w^{\epsi,\beta,N}_{12} p_1 \right\|  \left\llangle \psi,    \,\widehat{m_1}^2\,  q_2^\chi  \psi \right\rrangle^\frac{1}{2}} \\
&\weq{\ref{wcor}}{\lesssim}&   \|\Phi\|_{H^2(\R)}^2   \left\llangle \psi,   \,\widehat{m_1}^2 \,q_2^\chi  \psi \right\rrangle^\frac{1}{2} 
 \;\weq{\ref{lem:qs&N}}{ \lesssim}\;   \|\Phi\|_{H^2(\R)}^2 N^{\xi}\,\epsi\,   g(t)  \,.
\end{eqnarray*}

In the strongly confining case we have $b=0$ and instead estimate I by
\[
|{\rm I}|\leq \left|\left\llangle  \psi, p_1 p_2   w^{\epsi,\beta,N}_{12} \widehat m_1 p_1 q_2   \psi \right\rrangle  \right| \leq \|p_1 w^{\epsi,\beta,N}_{12} p_1\|\,\|\widehat m_1   q_2   \psi \|\lesssim \mu\,\|\Phi\|^2_{H^2(\R)}\,.
\]

\end{proof}

\begin{proof}[Proof of the bound for {\rm II}]
We start   with the case of moderate confinement.
Using again    Lemma~\ref{lem:weights} (c) we find that 
\begin{eqnarray}
|{\rm II}|&=&\left| \left\llangle  \psi, p_1 p_2 \,w^{\epsi,\beta,N}_{12}  \widehat{m_2}   \, q_1 q_2   \psi \right\rrangle  \right|  \;=\; \left|\left\llangle  \psi, p_1 p_2  (\widehat {\tau_2 m_2}   )^\frac{1}{2} w^{\epsi,\beta,N}_{12}  \widehat{m_2}    ^\frac{1}{2}  q_1 q_2  \psi \right\rrangle  \right|\nonumber\\
&= &\frac{1}{N} \Big|\sum_{j=2}^N  \left\llangle  \psi,   (\widehat{ \tau_2 m_2}  )^\frac{1}{2} p_1 p_j w^{\epsi,\beta,N}_{1j}  q_1 q_j  \widehat{m_2}^\frac{1}{2}   \psi \right\rrangle \Big|\nonumber\\
&  \lesssim  & \frac{1}{N} \Big\| \sum_{j=2}^N  q_j  w^{\epsi,\beta,N}_{1j} ( \widehat { \tau_2 m_2}   )^\frac{1}{2}  p_1 p_j \psi \Big\| \,\left\|  \widehat{m_2}^\frac{1}{2} q_1 \psi \right\|\, .  \label{equ:II.1}
\end{eqnarray}
The second factor of \eqref{equ:II.1} is easily estimated by
\[
\left\|  \widehat{m_2}^\frac{1}{2} q_1 \psi\right\|^2=\left\llangle \psi, \widehat{m_2}  \, q_1 \psi \right\rrangle  \weq{\ref{lem:weights}}{ =} \left\llangle \psi, \widehat{m_2}  \, \widehat{n}^2 \psi \right\rrangle =  \left\llangle \widehat n^\frac12 \psi, \widehat{  m_2n}\,   \widehat n^\frac12   \psi \right\rrangle \leq \| m_2n\|_\infty \left\| \widehat n^\frac12\,\psi\right\|^2
\lesssim \alpha_\xi\,.  
\]
The first factor of \eqref{equ:II.1} we split into a ``diagonal'' and an ``off-diagonal'' term and find
\begin{eqnarray}\lefteqn{\hspace{-2cm}
 \norm{\sum_{j=2}^N q_j w^{\epsi,\beta,N}_{1j}   ( \widehat { \tau_2 m_2}   )^\frac{1}{2} p_1 p_j \psi }^2
 = \sum_{j,l=2}^N \left\llangle \psi, p_1 p_l  \left( \widehat { \tau_2 m_2}   \right)^\frac{1}{2}  w^{\epsi,\beta,N}_{1l}  q_l q_j w^{\epsi,\beta,N}_{1j} 
 \left( \widehat { \tau_2 m_2}   \right)^\frac{1}{2}  p_1 p_j \psi  \right\rrangle}\nonumber\\
&\leq &\sum_{2 \leq  j < l \leq N} \left\llangle \psi, q_j p_1 p_l  \left( \widehat { \tau_2 m_2}   \right)^\frac{1}{2} w^{\epsi,\beta,N}_{1l}    w^{\epsi,\beta,N}_{1j} q_l  \left( \widehat { \tau_2 m_2}   \right)^\frac{1}{2}  p_1 p_j \psi \right\rrangle \nonumber\\
&&+\;(N-1) \norm{w^{\epsi,\beta,N}_{12}    p_1 p_2  \left( \widehat { \tau_2 m_2}   \right)^\frac{1}{2}\,\psi   }^2\, .\label{equ:II.3gp}
\end{eqnarray}
The first summand of \eqref{equ:II.3gp} is bounded by
\begin{eqnarray}\lefteqn{\hspace{-.5cm}
  (N-1)(N-2)\llangle \psi, q_2 p_1 p_3  ( \widehat { \tau_2 m_2}   )^\frac{1}{2}  w^{\epsi,\beta,N}_{13}    w^{\epsi,\beta,N}_{12} q_3 ( \widehat { \tau_2 m_2}   )^\frac{1}{2}  p_1 p_2 \psi  \rrangle }\nonumber\\
& \leq&  N^2 \norm { \sqrt{ w^{\epsi,\beta,N}_{13}}    \sqrt{ w^{\epsi,\beta,N}_{12}} q_3  ( \widehat { \tau_2 m_2}   )^\frac{1}{2}  p_1 p_2 \psi}^2\nonumber\\
&\leq& N^2 \norm{   \sqrt{ w^{\epsi,\beta,N}_{12}} p_2 \sqrt{ w^{\epsi,\beta,N}_{13}} p_1 ( \widehat { \tau_2 m_2}   )^\frac{1}{2} q_3 \psi } ^2  
 \leq  N^2 \norm{  \sqrt{ w^{\epsi,\beta,N}_{12}} p_2  }^4  \norm{  ( \widehat { \tau_2 m_2}   )^\frac{1}{2} \,q_3 \psi }^2\nonumber\\
&\weq{\ref{wcor},\ref{lem:weights}}{\lesssim} &N^2 \|\Phi\|_{H^2(\R)}^4 \norm{  ( \widehat { \tau_2 m_2}   )^\frac{1}{2}\, \widehat n \,\psi }^2 = N^2  \|\Phi\|_{H^2(\R)}^4  \left\llangle \widehat n^\frac12 \psi, \widehat{  \tau_2m_2n}\,   \widehat n^\frac12   \psi \right\rrangle \nonumber\\
&\lesssim &N^2 \|\Phi\|_{H^2(\R)}^4 \,\alpha_\xi
\,,  \label{equ:II.4gp}
\end{eqnarray}
where we used that $ \tau_2m_2n$ is bounded.
The second summand of \eqref{equ:II.3gp} is bounded by
\begin{eqnarray} \lefteqn{\hspace{-4cm}
 N \left\llangle  ( \widehat { \tau_2 m_2}   )^\frac{1}{2} \, \psi, p_1 p_2 (w^{\epsi,\beta,N}_{12})^2  p_1 p_2   ( \widehat { \tau_2 m_2}   )^\frac{1}{2}\,\psi  \right\rrangle  \leq  N \left\|p_1 (w^{\epsi,\beta,N}_{12})^2 p_1\right\| \left\| ( \widehat { \tau_2 m_2}   )^\frac{1}{2}\right\|^2}  \nonumber\\& \weq{\ref{wcor}}{ \lesssim } & N  \frac{\epsi^2}{\mu^{3}}\|\Phi\|_{H^2(\R)}^2   N^\xi  \,,   \label{equ:II.5gp}
\end{eqnarray}
since $\sup_{1\leq k \leq N }  m_2(k)    \; \leq \; N^{\xi}$.
Inserting the bounds \eqref{equ:II.4gp} and \eqref{equ:II.5gp} into \eqref{equ:II.3gp}, we obtain in continuation of \eqref{equ:II.1}
the desired bound,
\begin{eqnarray*}
|{\rm II}| &  \lesssim&   \Big ( \norm{\Phi}_{H^2(\R)}^2 \sqrt{\alpha_\xi} + N^{-\frac{1}{2}} \frac{\epsi}{\mu^{3/2}} \norm{\Phi}_{H^2(\R)} N^\frac{\xi}{2}\Big ) \sqrt{\alpha_\xi}\\
&=&  \norm{\Phi}_{H^2(\R)}^2  {\alpha_\xi} +  N^\frac{\xi}{2} \sqrt\frac{a}{\mu^3} \norm{\Phi}_{H^2(\R)} \sqrt{\alpha_\xi}
\;\leq\;\tfrac32 \norm{\Phi}_{H^2(\R)}^2  {\alpha_\xi} +  N^\xi  \frac{a}{\mu^3}\,.
\end{eqnarray*}
\end{proof}
In the strongly confining case we can easily estimate 
\begin{eqnarray*}
|{\rm II}|&\leq &\left|\left\llangle  \psi, p_1 p_2   w^{\epsi,\beta,N}_{12} \widehat m_2 q_1 q_2   \psi \right\rrangle  \right| \leq \|w^{\epsi,\beta,N}_{12} p_1\|\,\|\widehat m_2 q_1  q_2   \psi \|\lesssim \sqrt{\mu}\,\|\Phi\|_{H^2(\R)}\sqrt{\alpha_\xi}\\&\leq&  \|\Phi\|_{H^2(\R)}(\alpha_\xi+\mu)\,.
\end{eqnarray*}

\begin{proof}[Proof of the bound for {\rm III}]
The same  manipulations as before  yield
\begin{eqnarray}
 |\mathrm{III}|&=& \left|N\left\llangle \psi, p_1 q_2  \left[    w^{\epsi, \beta,N}_{12}-   b|\Phi|^2(x_1) , \widehat m\right] q_1q_2 \psi \right\rrangle\right| \nonumber\\&
 \weq{\ref{lem:weights}}{=} &
 \left|\left\llangle \psi, p_1 q_2  \left(   w^{\epsi, \beta,N}_{12}-   b|\Phi|^2(x_1) \right)  \widehat {m_1} \,q_1q_2 \psi \right\rrangle\right| \nonumber\\&
\leq&
 \left|\left\llangle \psi, p_1 q_2 \, w^{\epsi, \beta,N}_{12}\, \widehat {m_1} \,q_1q_2 \psi \right\rrangle\right| +  \left|\left\llangle \psi, p_1 q_2 \,  b|\Phi|^2(x_1) \, \widehat {m_1} \,q_1q_2 \psi \right\rrangle\right|\,. \label{equ:III.1gp}
\end{eqnarray}
The second summand of \eqref{equ:III.1gp} is easily bounded by
\[
 \left|\left\llangle \psi, p_1 q_2 \,  b|\Phi|^2(x_1) \, \widehat {m_1} \,q_1q_2 \psi \right\rrangle\right|
 \lesssim \|q_2\psi\| \,\|\widehat {m_1}q_1q_2\psi\|\lesssim \alpha_\xi\,.
\]
 
 For the first term of \eqref{equ:III.1gp} we use $q=q^\chi + p^\chi q^\Phi$ to obtain four terms
\begin{eqnarray}
|\llangle  \psi, p_1  q_2  w^{\epsi,\beta,N}_{12}  \, \widehat {m_1} \, q_1 q_2   \psi \rrangle  | &\leq&|\llangle  \psi, p_1  q_2^\chi  w^{\epsi,\beta,N}_{12}  \, \widehat {m_1} \, q_1 q_2   \psi \rrangle  |\;+\;|\llangle  \psi, p_1   p_2^\chi q_2^\Phi  w^{\epsi,\beta,N}_{12} \, \widehat {m_1} \, q_1 q_2^\chi   \psi \rrangle  |.
  \nonumber\\
&&+\;|\llangle  \psi, p_1   p_2^\chi q_2^\Phi  w^{\epsi,\beta,N}_{12} \, \widehat {m_1} \,  q_1^\chi q_2  \psi \rrangle  |\nonumber\\
&& +\;|\llangle  \psi, p_1  p_2^\chi q_2^\Phi  w^{\epsi,\beta,N}_{12}  \, \widehat {m_1} \, p_1^\chi q_1^\Phi p_2^\chi q_2^\Phi   \psi \rrangle  |\,.
 \label{equ:III.2gp}
\end{eqnarray}
All terms but the last are easy to handle.
The first term of \eqref{equ:III.2gp} can be estimated by
\begin{eqnarray}\lefteqn{\hspace{-1cm}
 |\llangle  \psi, p_1  q_2^\chi  w^{\epsi,\beta,N}_{12}  \, \widehat {m_1} \, q_1 q_2   \psi \rrangle  | \leq \norm{ q_2^\chi \psi} \norm{w^{\epsi,\beta,N}_{12}p_1}  \norm{ \, \widehat {m_1} \,q_1 q_2 \psi}}\nonumber\\
&\lesssim& \epsi g(t) \,\tfrac{\epsi}{\mu^{3/2}}  \norm{\Phi}_{H^2(\R)} \,\sqrt{\alpha_\xi}\leq g(t) \norm{\Phi}_{H^2(\R)} \left( \alpha_\xi + \tfrac{\epsi^4}{\mu^3}\right)\,, \label{equ:III.8gp}
\end{eqnarray}
where we used Lemmas~\ref{lem:weights}  and \ref{lem:qs&N} (b)  and Corollary~\ref{wcor}  in the second step.   
For the second (and completely analogous the third) term in  \eqref{equ:III.2gp} we find in the same way
\begin{eqnarray}\lefteqn{\hspace{-1cm}
|\llangle  \psi, p_1   p_2^\chi q_2^\Phi  w^{\epsi,\beta,N}_{12} \, \widehat {m_1} \,  q_1 q_2^\chi  \psi \rrangle  |=|\llangle  \psi, p_1   p_2^\chi q_2^\Phi  \, \widehat {\tau_1 m_1}^\frac12 \, w^{\epsi,\beta,N}_{12}  \, \widehat {m_1}^\frac12 \,q_1  q_2^\chi  \psi \rrangle  |}\nonumber\\
&\leq& \norm{ \widehat{\tau_1 m_1}^\frac12q_2  \psi} \left\|w^{\epsi,\beta,N}_{12}p_1\right\| \norm{ \, \widehat {m_1}^\frac12 \,q_1  q_2^\chi \psi}\nonumber\\
&\lesssim&\sqrt{\alpha_\xi} \,\tfrac{\epsi}{\mu^{3/2}} \,\norm{\Phi}_{H^2(\R)}\, \epsi g(t) \leq  g(t)\norm{\Phi}_{H^2(\R)} \left( \alpha_\xi + \tfrac{\epsi^4}{\mu^3}\right)\,,  \label{equ:III.7gp}
\end{eqnarray}
where we used
\begin{eqnarray*}
\norm{ \, \widehat {m_1}^\frac12 \,q_1  q_2^\chi \psi}^2&=& \left\llangle q_2^\chi\psi, \widehat {m_1}q_1\psi\right\rrangle \;= \;\frac{1}{N-1}\sum_{j=2}^N
\left\llangle q_j^\chi\psi, \widehat {m_1}q_1\psi\right\rrangle\\
&=& \frac{1}{N-1} 
\left\llangle \sum_{j=1}^Nq_j^\chi\psi, \widehat {m_1}q_1\psi\right\rrangle - \frac{1}{N-1} 
\left\llangle q_1^\chi\psi, \widehat {m_1}q_1\psi\right\rrangle\\
&=&\frac{1}{N-1} 
\left\llangle \sum_{j=1}^Nq_j^\chi\psi, \widehat {m_1}\widehat{n}^2\psi\right\rrangle -  
\left\llangle q_1^\chi\psi, \frac{\widehat {m_1}}{N-1}q_1^\chi\psi\right\rrangle\\
&\leq&
\frac{N}{N-1}  \|q_1^\chi\psi\|^2 +\|q_1^\chi\psi\|^2\;\lesssim\; \epsi^2  g(t)^2 \,.
\end{eqnarray*}

In the last term of \eqref{equ:III.2gp} we again split the interaction according to Lemma~\ref{R2Lemma}
\begin{eqnarray*}
|\llangle  \psi, p_1  p_2^\chi q_2^\Phi  w^{\epsi,\beta,N}_{12}  \, \widehat {m_1} \, p_1^\chi q_1^\Phi p_2^\chi q_2^\Phi   \psi \rrangle  | &=& |\llangle  \psi, p_1  p_2^\chi q_2^\Phi  w^{0}_{12}  \, \widehat {m_1} \, p_1^\chi q_1^\Phi p_2^\chi q_2^\Phi   \psi \rrangle  |\\ &&+\;|\llangle  \psi, p_1  p_2^\chi q_2^\Phi (T_1+T_2)  \, \widehat {m_1} \, p_1^\chi q_1^\Phi p_2^\chi q_2^\Phi   \psi \rrangle  | 
\end{eqnarray*}
and bound the second term with the help of Corollary~\ref{wcor} and Lemmas~\ref{lem:weights}  and \ref{lem:qs&N} (b),
\[
|\llangle  \psi, p_1  p_2^\chi q_2^\Phi (T_1+T_2)  \, \widehat {m_1} \, p_1^\chi q_1^\Phi p_2^\chi q_2^\Phi   \psi \rrangle  |\lesssim \tfrac{(\epsi+\mu)\epsi}{\mu^{3/2}} \|\Phi\|_{H^2(\R)}\sqrt{\alpha_\xi} \leq \|\Phi\|_{H^2(\R)}\left(\alpha_\xi +  \tfrac{(\epsi+\mu)^2\epsi^2}{\mu^{3}}
\right)\,.
\]
 For the leading term containing $w^0_{12}$ we have to use a different approach. Here we know that the potential only acts on the function $\chi$ in the confined directions. Thus,
 we  can replace 
\[
p_1  p_2^\chi   \, w^0_{12} \,   p_1^\chi   p_2^\chi  = p_1  p_2^\chi \,   \overline w^0_{12}    \,p_1^\chi   p_2^\chi
\]
with
\[
 \overline w^0(x_1-x_2): =  \frac{1}{\mu } \int_{\Omega_{\rm f}^2} \frac{\epsi^2}{\mu^2} \,  w\Big (\mu^{-1} \big(x_1-x_2,\epsi(y_1-y_2)\big) \Big) 
  |\chi(y_1)|^2 |\chi(y_2)|^2 \D y_1 \D y_2 \,.
\]
By inspection of the above formula on checks that $\| \overline w^0\|_{L^1(\R)} \lesssim 1$ 
and thus its anti-derivative   
\begin{align*}
\overline W^0(x):= \int_{-\infty}^{x'} \overline  w^0(x') \D x' \leq \left\| \overline w^0\right\|_{L^1(\R)} 
\end{align*}
remains bounded. Integration by parts therefore yields
\begin{eqnarray*}\lefteqn{
\left\llangle  \psi, p_1  p_2^\chi q_2^\Phi  w^0_{12}  \, \widehat {m_1} \, p_1^\chi q_1^\Phi p_2^\chi q_2^\Phi   \psi \right\rrangle = \left\llangle  \psi, p_1  p_2^\chi q_2^\Phi  \left(\tfrac{\partial}{\partial x_1} \overline W^0_{12}\right)  \, \widehat {m_1} \, p_1^\chi q_1^\Phi p_2^\chi q_2^\Phi   \psi \right\rrangle}\\
&=&  -\,\left\llangle  \psi,  \left(\tfrac{\partial}{\partial x_1} p_1 \right)  p_2^\chi q_2^\Phi  \overline W^0_{12} \, \widehat {m_1} \, p_1^\chi q_1^\Phi p_2^\chi q_2^\Phi   \psi \right\rrangle - 
\left\llangle  \psi, p_1  p_2^\chi q_2^\Phi \, \overline W^0_{12}   \,\tfrac{\partial}{\partial x_1}\, \widehat {m_1} \, p_1^\chi q_1^\Phi p_2^\chi q_2^\Phi   \psi \right\rrangle\,,
\end{eqnarray*}
where the first term is easily bounded by
\begin{eqnarray*} 
\left|\left\llangle  \psi,  \left(\tfrac{\partial}{\partial x_1} p_1 \right)  p_2^\chi q_2^\Phi  \overline W^0_{12} \, \widehat {m_1} \, p_1^\chi q_1^\Phi p_2^\chi q_2^\Phi   \psi \right\rrangle \right| &\leq & \left\|\tfrac{\partial}{\partial x_1} p_1\right\| \left\| q_2\psi\right\| 
\left\|\overline W^0_{12}\right\|_\infty \left\|\widehat {m_1} q_1q_2\psi\right\| \\
&\lesssim& \|\Phi\|_{H^1(\R)} \alpha_\xi\,.
\end{eqnarray*}
The second term is 
\begin{eqnarray*}  \lefteqn{\hspace{-1cm}
\left| \left\llangle  \psi, p_1  p_2^\chi q_2^\Phi \, \overline W^0_{12}   \,\tfrac{\partial}{\partial x_1}\, \widehat {m_1} \, p_1^\chi q_1^\Phi p_2^\chi q_2^\Phi   \psi \right\rrangle\right| =}\\&=& \left| \left\llangle  \psi, p_1  p_2^\chi q_2^\Phi \,\overline W^0_{12}  (p_1+q_1) q_2 \,\tfrac{\partial}{\partial x_1}\, q_1q_2 \widehat {m_1} \, p_1^\chi q_1^\Phi p_2^\chi q_2^\Phi   \psi \right\rrangle\right| 
 \\
&=&\left| \left\llangle  \psi, p_1  p_2^\chi q_2^\Phi \,\overline W^0_{12}   (p_1 \widehat {\tau_{1} m_1} + q_1\widehat {m_1} )  q_2 \,\tfrac{\partial}{\partial x_1}\,  p_1^\chi q_1^\Phi p_2^\chi q_2^\Phi   \psi \right\rrangle\right| \\
&\leq& \left\| q_2\psi\right\| 
\left\|\overline W^0_{12}\right\|_\infty  \left\|\left(\widehat {\tau_{1} m_1} +\widehat {m_1}\right)q_2 \, \tfrac{\partial}{\partial x_1}\,q_1\psi\right\|\\
& \lesssim&   \sqrt{\alpha_\xi}  \left( \|\Phi\|_{H^2(\R)}^\frac{3}{2}   \sqrt{  \alpha_\xi   +\frac{\mu}{\epsi} +\frac{a}{\mu^3}} +\epsi g(t)\right)\\& \lesssim &
 \|\Phi\|_{H^2(\R)}^\frac{3}{2}   \left(  \alpha_\xi   +\frac{\mu}{\epsi}+\frac{a}{\mu^3} \right) + g(t) (\alpha_\xi +\epsi^2)
\,,
\end{eqnarray*}
where we used Lemma~\ref{lem:energyestimate} and for $\ell=0,1$

\begin{eqnarray*}
\left\|\widehat {\tau_{\ell} m_1} q_2 \,q_1\,\tfrac{\partial}{\partial x_1}\,p_1^\chi q_1^\Phi1\psi\right\|^2 &=& \left\llangle 
q_1\,\tfrac{\partial}{\partial x_1}\,p_1^\chi q_1^\Phi\psi, q_2 \widehat {\tau_{\ell} m_1}^2 
q_1\,\tfrac{\partial}{\partial x_1}\,p_1^\chi q_1^\Phi\psi\right\rrangle\\&\leq& \left\llangle 
q_1\,\tfrac{\partial}{\partial x_1}\,p_1^\chi q_1^\Phi\psi,  \sum_{j=1}^N q_j\frac{\widehat {\tau_{\ell} m_1}^2}{N-1}
q_1\,\tfrac{\partial}{\partial x_1}\,p_1^\chi q_1^\Phi\psi\right\rrangle \\
&=&  \left\llangle 
q_1\,\tfrac{\partial}{\partial x_1}\,p_1^\chi q_1^\Phi\psi,  \sum_{k=0}^N \frac{  {\tau_{\ell} m_1(k) }^2}{N-1} \sum_{j=1}^N q_j P_k
q_1\,\tfrac{\partial}{\partial x_1}\,p_1^\chi q_1^\Phi\psi\right\rrangle \\
&=&  \left\llangle 
q_1\,\tfrac{\partial}{\partial x_1}\,p_1^\chi q_1^\Phi\psi,  \sum_{k=0}^N \frac{  {\tau_{\ell} m_1(k) }^2}{N-1}  k P_k
q_1\,\tfrac{\partial}{\partial x_1}\,p_1^\chi q_1^\Phi\psi\right\rrangle \\&\leq&  \left\|
q_1\,\tfrac{\partial}{\partial x_1}\,p_1^\chi q_1^\Phi\psi\right\|^2
\end{eqnarray*}
and
\begin{eqnarray*}
\left\| q_1\,\tfrac{\partial}{\partial x_1}\,p_1^\chi q_1^\Phi\psi\right\| &\leq& 
\left\|\left(q_1\,p_1^\chi  \tfrac{\partial}{\partial x_1} \,p_1^\chi  q_1^\Phi + q_1\,p_1^\chi  \theta'(x_1) L_1 q_1\right)\psi\right\| + \left\|\,p_1^\chi  \theta'(x_1) L_1 q_1\psi\right\| \\&=&
\left\|q_1\,p_1^\chi  \left(\tfrac{\partial}{\partial x_1} +  \theta'(x_1) L_1 \right) q_1\psi\right\| + \left\|\,p_1^\chi  \theta'(x_1) L_1 q_1^\chi\psi\right\| \\
&\leq &
\left\| \left(\tfrac{\partial}{\partial x_1} +  \theta'(x_1) L_1 \right) q_1\psi\right\| + \left\|\,p_1^\chi  \theta'(x_1) L_1\right\| \| q_1^\chi\psi\|\\
&\lesssim & \left( \|\Phi\|_{H^2(\R)}^3   \left(\alpha_\xi   +\frac{\mu}{\epsi} + \frac{a}{ \mu^3} \right)\right)^\frac12\;+\; \epsi g(t)\,.
\end{eqnarray*}

In the strongly confining case we find again
\begin{eqnarray*}
|{\rm III}|&\leq& \left|\left\llangle  \psi, p_1 q_2   w^{\epsi,\beta,N}_{12} \widehat m_1 q_1 q_2   \psi \right\rrangle  \right| \leq \|w^{\epsi,\beta,N}_{12} p_1\|\,\|\widehat m_1 q_1  q_2   \psi \|\lesssim \sqrt{\mu}\,\|\Phi\|_{H^2(\R)}\sqrt{\alpha_\xi}
\\
&\leq &\|\Phi\|_{H^2(\R)}(\alpha_\xi+\mu)\,.
\end{eqnarray*}

 \end{proof}

\begin{proof}[Proof of the bound for {\rm IV}]
For the first two summands in IV we expand the potential around $y_1=0$. The assumption A3 guarantees that in both cases the error is a bounded operator.
Therefore, we can write
\[
 \dot V(t,x_1,\epsi y_1)= \dot V(t,x_1,0)+ \epsi Q \qquad  V(t,x_1,\epsi y_1)=  V(t,x_1,0)+ \epsi \tilde Q
\]
with $\norm{Q} , \|\tilde Q \|   \leq C$. Thus we find 
\begin{eqnarray}
   \left|\left\llangle \psi, p_1  N  \left[   V(x_1,\epsi y_1)-V(x_1,0) , \widehat m\right] q_1 \psi \right\rrangle\right| 
&=&   \left|\left\llangle \psi, p_1   \epsi\, \tilde Q\, \widehat{m}_1 q_1 \psi \right\rrangle\right| \lesssim  \epsi  \norm{  \widehat{m}_1 q_1 \psi } \weq{\ref{lem:qs&N} }{\leq} \epsi\,. \label{easybound}
\end{eqnarray}
For the term containing $\dot V$  we first note that for $f \in L^\infty(\R)$
\begin{align}\label{equ:einteilchenop}
 \left|\left\llangle \psi, f(x_1) \psi \right\rrangle -\left \langle \Phi, f(x) \Phi \right\rangle\right| \lesssim \norm{f}_{L^\infty(\R)} \alpha_\xi.
\end{align}
Thus we can estimate 
\begin{eqnarray*}
 \left|\left\llangle \psi, \dot V(x_1,\epsi y_1) \psi \right\rrangle - \left\langle \Phi, \dot V(x_1,0) \Phi \right\rangle\right|
&\lesssim& \left|\left\llangle  \psi,  \dot V(x_1,0)\psi \right\rrangle - \left\langle \Phi, \dot V(x_1,0) \Phi \right\rangle\right| +\epsi\\
& \weq{ \eqref{equ:einteilchenop}}{ \lesssim} & \norm{\dot V(\cdot,0)}_{L^\infty(\R)} \alpha_\xi +\epsi\,.
\end{eqnarray*}
Equation \eqref{equ:einteilchenop} holds since
\begin{eqnarray*}
  |\llangle  \psi, f(x_1) \psi \rrangle  - \langle \Phi, f(x) \Phi \rangle|&\leq& |\llangle  \psi, p_1 f(x_1) p_1 \psi \rrangle  - \langle \Phi, f(x) \Phi \rangle|
+ |\llangle  \psi, q_1 f(x_1) p_1 \psi \rrangle| \\ && +\; |\llangle  \psi, p_1 f(x_1) q_1 \psi \rrangle| + |\llangle  \psi, q_1 f(x_1) q_1 \psi \rrangle |  \\
&\leq& \alpha_\xi\,\langle \Phi, f(x) \Phi \rangle+ 2 |\llangle  \psi, \widehat{\tau_1 n}^{1/2} p_1 f(x_1) \widehat n^{-1/2} q_1 \psi \rrangle | 
+ \norm{f}_{L^\infty(\R)} \alpha_\xi\\
&\weq{\ref{lem:weights}}{ \lesssim}& \norm{f}_{L^\infty(\R)} \alpha_\xi.
\end{eqnarray*}

For the ``twisting'' term we find
 \begin{eqnarray*}
   \big\llangle \psi, p_1    \left(   (\theta'(x_1)L_1)^2 + |\theta'(x_1)|^2\|L\chi\|^2  \right) \widehat{m}_1 \,q_1 \psi \big\rrangle
 &=&   \big\llangle \psi, p_1    \left(   (\theta'(x_1)L_1)^2 +  |\theta'(x_1)|^2\|L\chi\|^2  \right) q_1^\Phi p_1^\chi \widehat{m}_1 \,  \psi \big\rrangle
 \\&& +  \; \big\llangle \psi, p_1    \left(   (\theta'(x_1)L_1)^2 + |\theta'(x_1)|^2\|L\chi\|^2  \right)   q_1^\chi \widehat{m}_1 \,  \psi \big\rrangle\,.
 \end{eqnarray*}
 With
 \[
   \left\langle \chi, (\theta'(x ) L )^2 \chi\right\rangle_{L^2(\Omega_{\rm f})} = - |\theta'(x) |^2 \left\langle L  \chi, L   \chi\right\rangle   
    \]
 we see that the first term vanishes identically. 
 For the second term we find with Lemma~\ref{lem:qs&N}~(b) that
 \begin{eqnarray*}\lefteqn{\hspace{-1cm}
 \left|\big\llangle \psi, p_1    \left(   (\theta'(x_1)L_1)^2 + |\theta'(x_1)|^2\|L\chi\|^2  \right)   q_1^\chi \widehat{m}_1   \psi \big\rrangle\right|}\\
 &\leq&
\| \left(  (\theta'(x_1)L_1)^2 +  |\theta'(x_1)|^2\|L\chi\|^2  \right)p_1 \psi\|\, \|\widehat{m}_1\,q_1^\chi  \,  \psi\|\lesssim g(t) N^\xi \epsi\,.
 \end{eqnarray*}
 The remaining one-body terms are  
 \[
R^{(1)} = - \partial_x  \theta' (x)  L - \theta' (x) L  \partial_x  \;+ \;\left(V_{\rm bend}(r)+\frac{\kappa(x)^2}{4}
\right) \;-\; \epsi \,S^\epsi\,.
\]
With $\langle\chi,L\chi\rangle = 0$ it holds  that
 \[
 \llangle \psi, p_1      (\partial_{x_1}  \theta' (x_1)  L_1 + \theta' (x_1) L_1  \partial_{x_1})  q_1^\Phi p_1^\chi \widehat{m}_1   \psi \big\rrangle
 = 0
 \]
 and for the remaining term
 \[
| \llangle \psi, p_1         (\partial_{x_1}  \theta' (x_1)  L_1 + \theta' (x_1) L_1  \partial_{x_1})  q_1^\chi \widehat{m}_1   \psi \big\rrangle|
 \lesssim g(t) N^\xi \epsi
 \]
 as before. With
 \[
 V_{\rm bend} (r)  + \frac{\kappa(x)^2}{4\rho_\epsi(r)^2} =  - \epsi\,\frac{T_{\theta(x)}y\cdot\kappa(x)''}{2\rho_\epsi(r)^3} -
  \epsi^2\, \frac{5(  T_{\theta(x)}y\cdot\kappa'(x))^2}{4\rho_\epsi(r)^4} = \Or(\epsi)
 \]
 we can proceed as in \eqref{easybound} for this part. For the $S^\epsi$ term first note that
 \[
s^\epsi(r):=  \epsi^{-1} (\rho_\epsi^{-2}(r) - 1) = \frac{ 2 T_{\theta(x)}y\cdot \kappa(x) - \epsi (T_{\theta(x)}y\cdot \kappa(x))^2}{(1- \epsi\, T_{\theta(x)}y\cdot \kappa(x))^2}
 \]
 is uniformly bounded on $\Omega$ with all its derivatives. Hence
 \begin{eqnarray*}
\epsi  \big\llangle \psi, p_1   S^\epsi  \widehat{m}_1  q_1 \psi \big\rrangle&=& 
\epsi  \big\llangle \psi, p_1   \left(\partial_{x_1}  +\theta'(x_1)   L_1 \right)  s^\epsi(r_1)  \left(\partial_{x_1}  +\theta'(x_1)   L_1 \right) \widehat{m}_1  q_1 \psi \big\rrangle\\
&\leq & \epsi \|   \left(\partial_{x_1}  +\theta'(x_1)   L_1 \right)  s^\epsi(r_1)  \left(\partial_{x_1}  +\theta'(x_1)   L_1 \right)p_1\psi\| \,\|\widehat{m}_1  q_1 \psi\|\lesssim \epsi\,\|\Phi\|_{H^2(\R)} \,,
 \end{eqnarray*}
 concluding the bound for $\rm IV$.
 \end{proof}

\subsection{Proof of Lemma~\ref{lem:energyestimate}}\label{energylemmaproof}

The strategy is to control the expression in terms of the energy per particle. To this end we observe that
\begin{eqnarray*}\lefteqn{
 \left\| \left(\tfrac{\partial}{\partial x_1} +  \theta'(x_1) L_1 \right) q_1\psi\right\|^2=
 - \left\llangle q_1 \psi, \left(\tfrac{\partial}{\partial x_1} +  \theta'(x_1) L_1 \right)^2q_1\psi\right\rrangle}
 \\
 &\leq &
 \left\llangle q_1 \psi, \left(\left(\tfrac{\partial}{\partial x_1} +  \theta'(x_1) L_1 \right)^2 -  \tfrac{1}{\epsi^2}\Delta_{y_1} + \tfrac{1}{\epsi^2} V^\perp(y_1)-\tfrac{E_0}{\epsi^2}
 \right)q_1\psi\right\rrangle
 \\
 &\leq &
2 \left\llangle  q_1\psi,  \left( - \left(\tfrac{\partial}{\partial x_1} + \theta'(x_1) L_1
 \right)  (1 + \epsi s^\epsi(r_1)) \left(\tfrac{\partial}{\partial x_1} + \theta'(x_1) L_1
 \right)-  \tfrac{1}{\epsi^2}\Delta_{y_1}+ \tfrac{1}{\epsi^2} V^\perp(y_1)-\tfrac{E_0}{\epsi^2}
 \right)q_1\psi\right\rrangle\\
 &=:& 2 \left\llangle q_1  \psi, \, \tilde h_1\;q_1\psi\right\rrangle\,.
\end{eqnarray*}

Hence
we have
\begin{eqnarray*}
 \left\| \left(\tfrac{\partial}{\partial x_1} +  \theta'(x_1) L_1 \right) q_1\psi\right\|^2&\leq& 2\norm{\sqrt{ \tilde h_1} q_1 \psi }^2 \leq  \norm{\sqrt {\tilde h_1}(1-p_1 p_2 )\psi}^2 +  \norm{\sqrt {\tilde h_1}p_1q_2\psi}^2 \nonumber\\&\leq&\left\llangle \psi , (1-p_1p_2 ) \tilde h_1 (1-p_1p_2) \psi  \right\rrangle+ \langle \ph,\tilde h_1\ph\rangle  {\alpha_\xi}\,.
\end{eqnarray*}
Note that
\begin{eqnarray*}
 \langle \ph,\tilde h_1\ph\rangle&=&- \left\langle \ph,\,   \left(\tfrac{\partial}{\partial x} + \theta'(x) L_1
 \right)  (1 + \epsi s^\epsi(r)) \left(\tfrac{\partial}{\partial x} + \theta'(x) L
 \right)\ph\right\rangle\\
 &\leq &-2 \left\langle \ph,\,   \left(\tfrac{\partial}{\partial x} + \theta'(x) L
 \right)^2  \ph\right\rangle\;=\; 2 \left( \left\| \tfrac{\partial}{\partial x} \Phi\right\|^2 +  \left\| |\theta'(x)|^2 \|L\chi\|^2  \Phi\right\|^2\right)\\
 &\lesssim &\|\Phi\|^2_{H^1(\R)}\,.
\end{eqnarray*}
Then, after expanding and rearranging the energy difference 
\begin{eqnarray*}
E^\psi -E^\Phi   &=&  \tfrac{1}{N}\big\llangle \psi ,H(t)\,\psi  \big\rrangle  -    \tfrac{E_0}{\epsi^2}  - 
\Big\langle \Phi , \mathcal{E}^\Phi (t)\Phi  \Big\rangle_{L^2(\R)}\\
&=& \left\llangle \psi ,\left(\tilde h_1+ \tfrac12 w^{\epsi,\beta,N}_{12}+V(x_1,\epsi y_1) + V_{\rm bend}(r_1)\right)\psi  \right\rrangle\\
&&-\; \Big\langle \Phi , \left(-\tfrac{\partial^2}{\partial x^2} - \tfrac{\kappa(x)^2}{4} +   |\theta'(x)|^2 \,\|L\chi\|^2   + V(x,0)+  \tfrac{b}{2} |\Phi |^2 \right)  \Phi  \Big\rangle_{L^2(\R)}
\end{eqnarray*}
we arrive at 
\begin{eqnarray}\label{hpg}\lefteqn{
\left\llangle \psi , (1-p_1p_2 ) \tilde h_1 (1-p_1p_2) \psi  \right\rrangle
= E^\psi- E^\Phi \notag} \\
&&-\;\left(  \left\llangle \psi , p_1p_2 \tilde h_1 p_1p_2 \psi  \right\rrangle- \left\langle \varphi, - \tfrac{\partial^2}{\partial x^2} - \tfrac{1}{\epsi^2} (\Delta_y+E_0) + |\theta'(x)|^2 \|L\chi\|^2\varphi \right\rangle\right) \label{grad2}
\\
&&-\;\left\llangle \psi , (1-p_1p_2 )\tilde h_1 p_1p_2 \psi  \right\rrangle-\left\llangle \psi , p_1p_2 \tilde  h_1 (1-p_1p_2) \psi  \right\rrangle\label{grad3} \\
&&-\;\tfrac12\left(   \llangle  \psi, p_1 p_2 w^{\epsi,\beta,N}_{12} p_1 p_2 \psi  \rrangle - \langle \Phi,   b|\Phi|^2 \Phi \rangle\right)\label{grad4} \\
&&-\; \tfrac12 \left( \left\llangle  \psi,(1- p_1 p_2)w^{\epsi,\beta,N}_{12}p_1 p_2 \psi  \right\rrangle+ \left\llangle  \psi, p_1 p_2 w^{\epsi,\beta,N}_{12}(1- p_1 p_2) \psi  \right\rrangle \right) \label{grad5} \\
&& -\; \tfrac12 \left\llangle  \psi,(1- p_1 p_2) w^{\epsi,\beta,N}_{12}(1- p_1 p_2) \psi  \right\rrangle \label{grad6} \\
&& -\;\left( \left\llangle \psi , V(x_1,\epsi y_1) \psi  \right\rrangle - \left\langle \Phi, V(x,0) \Phi \right\rangle\right) 
+\;\left( \left\llangle \psi ,  \tfrac{\kappa^2(x_1)}{4} \psi  \right\rrangle - \left\langle \Phi, \tfrac{\kappa^2(x)}{4} \Phi \right\rangle\right) 
  \label{grad7}\\
  &&-\; \left\llangle \psi ,\left(V_{\rm bend}(r_1)+\tfrac{\kappa(x_1)^2}{4}
\right) \psi  \right\rrangle\label{grad8}\,.
\end{eqnarray}
We will estimate each line separately.
For \eqref{grad2} we find
\begin{eqnarray*}
\eqref{grad2}&\leq& \left|\left\llangle \psi , p_1p_2  \tilde h_1 p_1p_2 \psi  \right\rrangle - \left\langle \varphi, \tilde h_1 \varphi \right\rangle \right|
+\epsi \langle \ph,\tilde h_1\ph\rangle  
\\&=&
\left|\left\langle \varphi, \tilde h_1 \varphi \right\rangle\left\llangle \psi , p_1p_2  \psi  \right\rrangle- \left\langle \varphi, \tilde h_1 \varphi \right\rangle \right| 
+\epsi \langle \ph,\tilde h_1\ph\rangle  \\
&= &\left|\left\langle \varphi, \tilde h_1 \varphi \right\rangle \left\llangle \psi , (1-p_1p_2 )\psi  \right\rrangle\right|+\epsi \langle \ph,\tilde h_1\ph\rangle  \\
&=& \langle \ph,\tilde h_1\ph\rangle\left( \left|\left\llangle \psi , (p_1q_2 +q_1p_2+ q_1q_2) \psi  \right\rrangle\right|+\epsi\right)\; \weq{\ref{lem:weights}}{  \lesssim }\; \|\Phi\|^2_{H^1(\R)}( \alpha_\xi+\epsi) \,,
\end{eqnarray*}
and \eqref{grad3}
is bounded in absolute value by 
\begin{eqnarray*}
 |\eqref{grad3}| &\leq& 2  \left|\left\llangle \psi, (1-p_1p_2 ) \tilde h_1 p_1p_2 \psi \right\rrangle\right| 
 =2  \left|\left\llangle \psi, q_1 \tilde h_1 p_1p_2 \psi \right\rrangle\right| =
 2\left|\left\llangle \psi, q_1 \widehat n^{-\frac12} \tilde h_1  \widehat{\tau_1n}^\frac12 p_1p_2 \psi \right\rrangle\right|
\\&\leq & 2 \left\|  \widehat n^{-\frac12}q_1\psi\right\|
 \left\| \tilde h_1 p_1\right\| \left\| \widehat{\tau_1n}^\frac 12 \psi\right\| 
\lesssim  \sqrt{\alpha_\xi} \norm{ \Phi}_{H^2(\R)}\sqrt{ \alpha_\xi + \tfrac{1}{\sqrt N}}\lesssim \norm{ \Phi}_{H^2(\R)} \left(\alpha_\xi + \tfrac{1}{\sqrt N}\right)\,.
\end{eqnarray*}
For \eqref{grad4} we first note that
\[
 \left|\left\langle \Phi,     b|\Phi|^2  \Phi \right\rangle- \left\langle \psi,  p_1 p_2    b|\Phi|^2_1  p_1 p_2 \psi \right\rangle \right| 
= \left|\left\langle \Phi,  b|\Phi|^2  \Phi \right\rangle\right|   \left|\left\langle \psi, (1-p_1p_2) \psi \right\rangle\right|  \lesssim  \norm{\Phi}_{L^\infty(\R)}^2 \alpha_\xi\,.
\]
Hence,
\begin{eqnarray*}
| \eqref{grad4}|&\leq & \left|\left\llangle \psi,  p_1 p_2 \left(b|\Phi|^2-  w^{\epsi,\beta,N}_{12}\right)  p_1 p_2 \psi \right\rrangle\right|    + \norm{\Phi}_{L^\infty(\R)}^2 \alpha_\xi\\
&\leq& \left|\left\llangle \psi,  p_1 p_2 \left(b|\Phi|^2-  w^0_{12}\right)  p_1 p_2 \psi \right\rrangle\right| +\|(T_1+T_2)p_1\|   + \norm{\Phi}_{L^\infty(\R)}^2 \alpha_\xi\\
&\stackrel{\eqref{faltest}}{\lesssim}    &\frac{\mu}{\epsi}\,\|\nabla |\Phi|^2\|_{L^2(\R)}\norm{\Phi}_{L^\infty(\R)}+ 
\frac{\epsi(\epsi+\mu)}{\mu^{3/2}}\|\Phi\|_{H^2(\R)}
+\norm{\Phi}_{L^\infty(\R)}^2 \alpha_\xi\,.
\end{eqnarray*}
For \eqref{grad5} we have that
\begin{eqnarray}
|\eqref{grad5}|& \leq & 2
\left|\left\llangle  \psi, p_1 p_2 w^{\epsi,\beta,N}_{12}(1- p_1 p_2) \psi \right\rrangle\right|  = \left|\left\llangle  \psi, p_1 p_2 w^{\epsi,\beta,N}_{12}(q_1p_2+ p_1q_2+ q_1q_2) \psi \right\rrangle\right| \nonumber\\
&\leq& 2\left| \left\llangle  \psi, p_1 p_2 w^{\epsi,\beta,N}_{12} q_1p_2 \psi \right\rrangle\right|+\left| \left\llangle  \psi, p_1 p_2 w^{\epsi,\beta,N}_{12} q_1q_2 \psi \right\rrangle \right|\,.\label{grad5s1}
\end{eqnarray}
The first summand in \eqref{grad5s1} is bounded by 
\begin{eqnarray*}
\left| \left\llangle  \psi, p_1 p_2 w^{\epsi,\beta,N}_{12} q_1p_2 \psi \right\rrangle\right|&=& \left| \left\llangle  \psi, p_1 p_2 \,\widehat {\tau_1 n}^\frac{1}{2}\,w^{\epsi,\beta,N }_{12}\, \widehat n^{-\frac{1}{2}}  \,  q_1 p_2  \psi \right\rrangle\right|\\&\leq &\norm{p_2 w^{\epsi,\beta,N }_{12} p_2} \norm{\widehat {\tau_1 n}^\frac{1}{2} \psi} \norm{\widehat n^{-\frac{1}{2}}   q_1 \psi}\lesssim \norm{\Phi}_{H^2(\R)}^2 \left(\alpha_\xi+\tfrac{1}{\sqrt N}\right) \,.
\end{eqnarray*}
For the second summand in \eqref{grad5s1} we first use symmetry to write
\begin{eqnarray*}
\left| \left\llangle  \psi, p_1 p_2 w^{\epsi,\beta,N}_{12} q_1q_2 \psi \right\rrangle \right|&=& \frac{1}{N-1} \left|\sum_{j=2}^N  \left\llangle  \psi,   p_1 p_j w^{\epsi,\beta,N}_{1j}  q_1 q_j    \psi \right\rrangle \right|\\&\leq& \frac{\norm{    q_1 \psi}}{N-1} \norm{\sum_{j=2}^N  q_j  w^{\epsi,\beta,N}_{1j}   p_1 p_j \psi }\leq\frac{\sqrt{\alpha_\xi }}{N-1} \norm{\sum_{j=2}^N  q_j  w^{\epsi,\beta,N}_{1j}   p_1 p_j \psi } \label{equ:II.1g}.
\end{eqnarray*}
Now the second factor can be  split into a ``diagonal''   and an ``off-diagonal'' term,
\begin{eqnarray*}\lefteqn{\hspace{-1cm}
  \norm{\sum_{j=2}^N  q_j w^{\epsi,\beta,N}_{12}    p_1 p_j  \psi }^2
=  \sum_{j,k=2}^N \left\llangle \psi, p_1 p_l   w^{\epsi,\beta,N}_{1l}  q_l q_j w^{\epsi,\beta,N}_{1j} 
   p_1 p_j \psi  \right\rrangle }\\
& \leq & \hspace{-10pt}\sum_{2 \leq  j < k \leq N} \left\llangle \psi,   q_j  p_1 p_l   w^{\epsi,\beta,N}_{1l}    w^{\epsi,\beta,N}_{1j}   q_l  p_1 p_j \psi  \right\rrangle+ (N-1) \norm{w^{\epsi,\beta,N}_{12}  p_1 p_2 \psi   }^2 . 
  \label{equ:II.3g}
\end{eqnarray*}
The ``off-diagonal'' term is bounded by
\begin{eqnarray*}\lefteqn{\hspace{-3cm}
  (N-1)(N-2)\left\llangle \psi,q_2 p_1 p_3    w^{\epsi,\beta,N}_{13}    w^{\epsi,\beta,N}_{12} q_3  p_1 p_2 \psi  \right\rrangle 
 \leq  N^2 \norm{   \sqrt{ w^{\epsi,\beta,N}_{12}} p_2 \sqrt{ w^{\epsi,\beta,N}_{13}} p_1 q_3 \psi } ^2} \\
&\leq &N^2 \norm{  \sqrt{ w^{\epsi,\beta,N}_{12}} p_2  }^4  \norm{ q_3 \psi }^2\;\weq{\ref{wcor}}{ \lesssim} \; N^2 \norm{\Phi}_{H^2(\R)}^4 \alpha_\xi  \label{equ:II.4g}.
\end{eqnarray*}
The ``diagonal'' term is bounded by
\begin{eqnarray*}
 N \left\llangle \psi, p_1 p_2   \left(w^{\epsi,\beta,N}_{12}\right)^2  p_1 p_2 \psi  \right\rrangle 
&\leq& N   \norm{p_1 (w^{\epsi,\beta,N}_{12})^2 p_1 }   \weq{\ref{wcor}}{ \leq}      \tfrac{N\epsi^2}{\mu^3}  \norm{\Phi}^2_{H^2(\R)}  \label{equ:II.5g}
\end{eqnarray*}
and we conclude that the second summand of \eqref{grad5s1} is bounded by
\begin{eqnarray*}
\left| \left\llangle  \psi, p_1 p_2 w^{\epsi,\beta,N}_{12} q_1q_2 \psi \right\rrangle \right| &\lesssim &\frac{\sqrt{ \alpha_\xi }}{N}  \sqrt{ N^2 \norm{\Phi}^4_{H^2(\R)} \alpha_\xi +  \tfrac{N\epsi^2}{\mu^3}\norm{\Phi}^2_{H^2(\R)}  }\\& \leq&
\norm{\Phi}^2_{H^2(\R)} \alpha_\xi + \tfrac{\sqrt{\alpha_\xi}}{\sqrt{N}}\sqrt{\tfrac{\epsi^2}{\mu^3}}\norm{\Phi}_{H^2(\R)}
\lesssim \norm{\Phi}^2_{H^2(\R)} \alpha_\xi+ \tfrac{\epsi^2}{N\mu^3}\,.
\end{eqnarray*}
In summary we thus have that
\[
|\eqref{grad5}|\;\lesssim \;\norm{\Phi}^2_{H^2(\R)} \left( \alpha_\xi+\tfrac{1}{\sqrt{N}}\right) + \frac{\epsi^2}{N\mu^3}\,.
\]
Since 
the interaction is non-negative, we have $\eqref{grad6}\leq 0$. 
With the same arguments as used in the proof of Proposition~\ref{lem:3termeg} part IV we find
\[
|\eqref{grad7}|\; \lesssim \;   \alpha_\xi +\epsi\,,
\]
and obviously $|\eqref{grad8}|\lesssim \epsi$.
In summary we thus showed
 \begin{eqnarray*}
 \left\| \left(\tfrac{\partial}{\partial x_1} +  \theta'(x_1) L_1 \right) q_1\psi\right\|^2 &\lesssim&   \|\Phi\|_{H^2(\R)}^3   \left(\alpha_\xi + \frac{1}{\sqrt{N}} + \epsi +\frac{\mu}{\epsi}+\sqrt{\mu}+ \frac{a}{ \mu^3}
 \right)
  \end{eqnarray*}
 and with 
\[
\epsi\lesssim \frac{\mu}{\epsi}\,,\quad \frac{1}{\sqrt{N}}\lesssim \frac{\mu}{\epsi}
\,,\quad \sqrt{\mu}\lesssim  \frac{\mu}{\epsi}\,,
\]
which holds for  moderate confinement, the statement of the lemma follows.

\begin{appendix}

\section{Well-posedness of the dynamical equations }\label{app:regsol}

The Hamiltonian $H_{\tube_\epsi}(t)$ given in  \eqref{hamilton1} is self-adjoint on $H^2(\tube_\epsi^N)\cap H^1_0(\tube_\epsi^N)$ for every $t\in\R$, since the potentials $V$ and $w$ are bounded by assumptions {\bf A2} and {\bf A3}. 
Hence $(U_\epsi)^{\otimes N} H_{\tube_\epsi}(t)(U_\epsi^*)^{\otimes N} + \sum_{i=1}^N \tfrac{1}{\epsi^2} V^\perp(y_i)$ is self-adjoint on $U_\epsi H^2(\tube_\epsi^N)\cap U_\epsi  H^1_0(\tube_\epsi^N)= H^2(\Omega^N)\cap H^1_0(\Omega^N)$, as $\sum_{i=1}^N \tfrac{1}{\epsi^2} V^\perp(y_i)$ is relatively bounded with respect to $(U_\epsi)^{\otimes N} H_{\tube_\epsi}(t)(U_\epsi^*)^{\otimes N} $ with relative bound smaller than one.
Finally  
 $t\mapsto V(t)\in \mathcal{L}(L^2)$ is continuous, so $H(t)$ generates an strongly continuous evolution family $U(t,0)$ such that for $\psi_0\in H^2(\Omega^N)\cap H^1_0(\Omega^N)$ the map $t\mapsto U(t,0)\psi_0$ satisfies the time-dependent Schr\"odinger equation.

Although the questions of well-posedness, global existence and conservation laws for the NLS  equation in our setting are well understood, we couldn't find a reference for global existence of $H^2$-solutions to \eqref{equ:grosspqwg} with time-dependent potential. We thus briefly comment on this point. The standard contraction argument (see e.g.\ Proposition~3.8 \cite{Tao06}) gives unique local existence of $H^s$-solutions $\Phi(t)$ for all $\frac12<s\leq 4$, since  under the hypotheses {\bf A1} and {\bf A3} on the external potential and the waveguide all potentials appearing in \eqref{equ:grosspqwg} are $C^4_{\rm b}$. Moreover, $\|\Phi(t)\|_{L^2} = \|\Phi(0)\|_{L^2}$ and the solution map $\Phi(0)\mapsto \Phi(t)$ is continuous in $H^s$. See  \cite{Sp14} for the details of this argument in the case of time-dependent potentials.

In order to show also global existence, assume without loss of generality (we can always add a real constant to the potential) that 
\[
 \inf_{t,x\in\R}   \left(  -\tfrac{\kappa^2(x)}{4} +   |\theta'(x)|^2 \,\|L\chi\|^2  + V(t,x,0) \right) \geq 0
 \]
  and recall the definition $E^{\Phi(t)}(t) := \left\langle\Phi(t), \mathcal{E}^{\Phi(t)}(t)\Phi(t)\right\rangle$ in \eqref{equ:enggross2}. Then   for $\Phi(t)\in H^2$ the map $t\mapsto E^{\Phi(t)}(t)$ is differentiable and we have
\begin{eqnarray*}
\| \Phi(t)\|^2_{H^1}   &\leq& E^{\Phi (t)}(t)  +1 \;=\; E^{\Phi (0)}(0) +1+ \int_0^t \tfrac{\D}{\D s}  E^{\Phi(s) }(s) \,\D s \\&=&E^{\Phi(0) }(0) +1+\int_0^t \left\langle \Phi(s), \dot V(s,\cdot,0) \Phi(s)\right\rangle\,\D s\\
&\leq & C\| \Phi(0)\|^2_{H^1}  +  \|\Phi(0)\|^2 \int_0^t \|\dot V(s,\cdot,0)\|_{L^\infty}\,\D s\,,
\end{eqnarray*}
which, by continuity of the solution map, extends to $\Phi(t)\in H^1$.
Hence     $\|\Phi(t)\|_{L^\infty} \leq \|\Phi(t)\|_{H^1} $ cannot blow up in finite time, which implies global existence of $H^1$-solutions. 

To control also the $H^2$-norm, first note that with
\begin{eqnarray*}
\| \mathcal{E}^{\Phi(t)}(t)\Phi(t)\|^2 &=& \Big\langle \Phi (t), \left(-\tfrac{\partial^2}{\partial x^2} - \tfrac{\kappa ^2}{4} +   |\theta' |^2 \,\|L\chi\|^2   + V(t,\cdot,0)+  \tfrac{b}{2} |\Phi (t)|^2 \right)^2  \Phi (t) \Big\rangle\\
&\geq& \left\|\tfrac{\partial^2}{\partial x^2} \Phi(t)\right\|^2 +  2\Re \Big\langle \Phi(t), \tfrac{\partial^2}{\partial x^2} \Big(\underbrace{-\tfrac{\kappa^2 }{4} +   |\theta' |^2 \,\|L\chi\|^2  + V(t,\cdot,0)}_{=:f(t,x)} + \tfrac{b}{2}|\Phi(t)|^2\Big) \Phi(t)\Big\rangle 
\end{eqnarray*}
and 
\begin{eqnarray*}
\big| \big\langle \Phi(t), \tfrac{\partial^2}{\partial x^2} \big(f+ \tfrac{b}{2}|\Phi(t)|^2\big) \Phi(t)\big\rangle \big|
&\leq &\big| \big\langle \Phi'(t),  \big(f+ b|\Phi(t)|^2\big) \Phi'(t)\big\rangle \big| 
+ \big| \big\langle \Phi'(t),  \tfrac{b}{2} \overline{\Phi'(t)}   \Phi(t)^2\big\rangle \big| 
+ \big| \big\langle \Phi'(t),   f'  \Phi(t)\big\rangle \big| \\
&\leq &     \|\Phi(t)\|^2_{H^1} \left( C +b \|\Phi(t)\|^2_{L^\infty}\right) + \tfrac{b}{2}   \|\Phi(t)\|^2_{H^1}   \|\Phi(t)\|^2_{L^\infty}
+ C \|\Phi(t)\|_{H^1} \\
&\leq&  C_1\|\Phi(t)\|^4_{H^1} 
\end{eqnarray*}
for some constant $C_1\in \R$ we have
\begin{eqnarray*}
\left\|\tfrac{\partial^2}{\partial x^2} \Phi(t)\right\|^2 
&\leq& \| \mathcal{E}^{\Phi(t)}(t)\Phi(t)\|^2 + 2 C_1\|\Phi(t)\|^4_{H^1} \,.
\end{eqnarray*}
 Moreover, for $\Phi(t)\in H^4(\R)$ we have 
\begin{eqnarray*}
 \big\| \mathcal{E}^{\Phi(t)}(t)\Phi(t)\big\|^2+1  &=&  \big\| \mathcal{E}^{\Phi(0)}(0)\Phi(0)\big\|^2+1 + \int_0^t \tfrac{\D}{\D s}  \left\langle  \mathcal{E}^{\Phi(s)}(s)\Phi(s),  \mathcal{E}^{\Phi(s)}(s)\Phi(s)\right\rangle
     \,\D s \\&=& \big\| \mathcal{E}^{\Phi(0)}(0)\Phi(0)\big\|^2 +1+ 2 \int_0^t  \Re\left\langle \dot V (s,x,0) \Phi(s),  \mathcal{E}^{\Phi(s)}(s)\Phi(s)\right\rangle
     \D s\\
     &\leq& C_2\|\Phi(0)\|_{H^2}^2   + C_3 \int_0^t  \big\| \mathcal{E}^{\Phi(s)}(s)\Phi(s)\big\| \,\D s\\
     &\leq& C_2\|\Phi(0)\|_{H^2}^2    + C_3 \int_0^t  \left( \big\| \mathcal{E}^{\Phi(s)}(s)\Phi(s)\big\|^2+1\right) \,\D s\,.
\end{eqnarray*}
An application of the Gr\"onwall inequality yields a bound of $ \left\| \mathcal{E}^{\Phi(t)}(t)\Phi(t)\right\|^2$ in terms of $ \|\Phi(0)\|_{H^2}^2$,
which, again by continuity of the solution map, extends to $\Phi(t)\in H^2(\R)$.
Hence the $H^2$-norm of $\Phi(t)$ remains bounded on bounded intervals in time.

\end{appendix}

\bibliographystyle{alphanum}

\begin{thebibliography}{99}

\bibitem{AGT}
R.\ Adami, F.\ Golse,  and A.\ Teta.
\newblock Rigorous derivation of the cubic NLS in dimension one.
\newblock {\em J.\ Stat.\ Phys.} 127(6): 1193--1220, 2007.

\bibitem{abdmehschweis05}
N.~Ben~Abdallah, F.~M\'ehats, C.~Schmeiser, and R.~Weish\"aupl.
\newblock The nonlinear Schr\"odinger equation with a strongly anisotropic
  harmonic potential.
\newblock {\em SIAM Journal on Mathematical Analysis}, 37(1):189--199, 2005.

\bibitem{BenOliSch12}
N.\ Benedikter, G.\ De~Oliveira, and B.\ Schlein.
\newblock Quantitative derivation of the Gross-Pitaevskii equation.
\newblock {\em Communications on Pure and Applied Mathematics}, 68(8):1399--1482, 2015.

\bibitem{BenPorSch15}
N.\ Benedikter, M.\ Porta, and B.\ Schlein..
\newblock {Effective Evolution Equations from Quantum Dynamics}.
\newblock ArXiv:1502.02498, 2015.	 


\bibitem{CheHol13}
X.~Chen and J.~Holmer.
\newblock On the rigorous derivation of the 2d cubic nonlinear Schr\"odinger
  equation from 3d quantum many-body dynamics.
\newblock {\em Archive for Rational Mechanics and Analysis}, 210(3):909--954,
  2013.

\bibitem{CheHol14}
X.~{Chen} and J.~{Holmer}.
\newblock {Focusing quantum many-body dynamics II: the rigorous derivation of
  the 1d focusing cubic nonlinear Schr\"odinger equation from 3d}.
\newblock ArXiv:1407.8457, 2014.

   
\bibitem{ErdSchYau07}
L.~Erd{\H{o}}s, B.~Schlein, and H.-T.\ Yau.
\newblock Derivation of the cubic non-linear Schr\"odinger equation from quantum dynamics of many-body systems.
\newblock {\em Invent.\ Math.}, 167:515--614, 2007.

\bibitem{ErdYau01}
L.~Erd{\H{o}}s and H.-T. Yau.
\newblock Derivation of the nonlinear {S}chr\"odinger equation from a many body
  {C}oulomb system.
\newblock {\em Adv.\ Theor.\ Math.\ Phys.}, 5(6):1169--1205, 2001.


\bibitem{FZ}
J.\ Fort\'agh and C.\ Zimmermann.
\newblock Magnetic microtraps for ultracold atoms.
\newblock {\em Rev.\ Mod.\ Phys.}, 79(1):235--289, 2007.


 
\bibitem{Gol13}
F.\ Golse.
\newblock {On the Dynamics of Large Particle Systems in the Mean Field Limit}.
\newblock ArXiv:1301.5494, 2013.


\bibitem{GoeVogKet01}
A.~G\"orlitz, J.~M.\ Vogels, A.~E.\ Leanhardt, C.~Raman, T.~L.\ Gustavson, J.~R.\
  Abo-Shaeer, A.~P.\ Chikkatur, S.~Gupta, S.~Inouye, T.~Rosenband, and
  W.~Ketterle.
\newblock Realization of Bose-Einstein condensates in lower dimensions.
\newblock {\em Phys.\ Rev.\ Lett.}, 87:130402,   2001.


\bibitem{GrMa13}
M.\ Grillakis and M.\ Machedon.
\newblock Pair excitations and the mean field approximation of interacting Bosons.
\newblock {\em Commun.\ Math.\ Phys.},   324:601--636, 2013. 


\bibitem{HaLaTe14}
S.\ Haag, J.\ Lampart, and S.\ Teufel.
\newblock Generalised quantum waveguides.
\newblock {\em Annales Henri Poincar{\'e}},   16:2535--2568, 2015.  

\bibitem{HeRy09}
K.~Henderson, C.~Ryu, C.~MacCormick, and M.~G.\ Boshier.
\newblock Experimental demonstration of painting arbitrary and dynamic
  potentials for Bose-Einstein condensates.
\newblock {\em New Journal of Physics}, 11(4):043030, 2009.



\bibitem{KnoPic09}
A.~Knowles and P.~Pickl.
\newblock Mean-field dynamics: singular potentials and rate of convergence.
\newblock {\em Comm.\ Math.\ Phys.}, 298(1):101--138, 2010.


\bibitem{Kre08}
D.\ Krej\v{c}i\v{r}{\'i}k.
\newblock Twisting versus bending in quantum waveguides.
\newblock {\em  Analysis on Graphs and its Applications: Proceedings of the Symposium on Pure Mathematics, American Mathematical Society}, 617--636, 2008.

\bibitem{LT} J.\ Lampart and S.\ Teufel.
\newblock The adiabatic limit of Schr\"odinger operators on fibre bundles.
\newblock To appear in {\em Mathematische Annalen} (see also  ArXiv:1402.0382, 2014).

\bibitem{LewNamRou14}
M.~{Lewin}, P.T.~{Nam}, and N.~{Rougerie}.
\newblock {The mean-field approximation and the non-linear Schr\"odinger
  functional for trapped Bose gases}.
\newblock {\em Trans.\ Amer.\ Math.\ Soc.}, 368:6131-6157, 2016.



\bibitem{LieSeiSolYng05}
E.\ H.\ Lieb, R.~Seiringer, J.\ P.\ Solovej, and J.~Yngvason.
\newblock {\em The mathematics of the Bose gas and its condensation}, Volume~34
  of {\em Oberwolfach Seminars}.
\newblock Birkh\"auser, 2005.

\bibitem{LieSeiYng03}
E.\ H.\ Lieb, R.~Seiringer, and J.~Yngvason.
\newblock One-dimensional behaviour of dilute, trapped Bose gases.
\newblock  {\em Comm.\ Math.\ Phys.}, 244(2):347--393, 2004.


\bibitem{MehRay15}
F.\ M\'ehats and N.\ Raymond.
\newblock  Strong confinement limit for the nonlinear Schr\"odinger equation constrained on a curve.
\newblock ArXiv:1412.1049, 2014.


\bibitem{NamRouSei15}
P.T.\ {Nam}, N.\ Rougerie, and R.\ Seiringer.
\newblock Ground states of large bosonic systems: The Gross-Pitaevskii limit revisited.
\newblock ArXiv:1503.07061, 2015.


\bibitem{NaNa15}
P.T.\ {Nam}  and M.\ {Napi{\'o}rkowski}.
\newblock Bogoliubov correction to the mean-field dynamics of interacting bosons,
\newblock ArXiv:1509.04631, 2015.



\bibitem{Pic08}
P.~{Pickl}.
\newblock {On the time dependent Gross-Pitaevskii- and Hartree equation}.
\newblock ArXiv:0808.1178, 2008.


\bibitem{Pic10a}
P.~{Pickl}.
\newblock {Derivation of the time dependent Gross-Pitaevskii equation with
   external fields}.
\newblock {\em Rev.\ Math.\ Phys.}, 27:1550003, 2015. (see also  arXiv:1001.4894).


\bibitem{Pic10b}
P.\ Pickl.
\newblock
Derivation of the time dependent Gross-Pitaevskii equation without positivity condition on the interaction.
\newblock {\em J.\ Stat.\ Phys.} 140(1):76--89, 2010.


\bibitem{Pic11}
P.~Pickl.
\newblock A simple derivation of mean field limits for quantum systems.
\newblock {\em Lett.\   Math.\ Phys.}, 97(2):151--164, 2011.

 
\bibitem{RodSch07}
I.~Rodnianski and B.~Schlein.
\newblock Quantum fluctuations and rate of convergence towards mean field
  dynamics.
\newblock {\em Comm.\ Math.\ Phys.}, 291(1):31--61, 2009.

 
 
\bibitem{Rou15}
N.\ Rougerie.
\newblock {De finetti theorems, mean-field limits and bose-Einstein condensation}.
\newblock ArXiv:1506.05263, 2015.


 
\bibitem{SchYng06}
K.\ Schnee and J.\ Yngvason.
\newblock Bosons in Disc-Shaped Traps: From 3D to 2D.
\newblock {\em Comm.\ Math.\ Phys.}, 269(3):659--691, 2006.


\bibitem{Sch08}
B.\ Schlein.
\newblock {Derivation of Effective Evolution Equations from Microscopic Quantum Dynamics}.
\newblock ArXiv:0807.4307, 2008.

 
\bibitem{Sp14}
C.\ Sparber.
\newblock Weakly nonlinear time-adiabatic theory.
\newblock {\em Annales Henri Poincar\'e}, Online First, 2015  (see also arXiv:1411.0335).

\bibitem{Tao06}
T.~Tao.
\newblock {\em Nonlinear dispersive equations. Local and global analysis. CBMS
  Regional Conference Series in Mathematics, 106}, Volume 200,
\newblock 2006.

\bibitem{WT} J.\ Wachsmuth and S.\ Teufel.
\newblock
Effective Hamiltonians for constrained quantum systems.
\newblock {\em Memoirs of the AMS} 1083 (2013).


\end{thebibliography}
\newcommand{\etalchar}[1]{$^{#1}$}

\end{document}